\documentclass[11pt, reqno]{amsart}
  \usepackage{amsfonts,latexsym, enumerate}
  \usepackage{amsmath}
  \usepackage{amscd}
  \usepackage{float,amsmath,amssymb,mathrsfs,bm,multirow,graphics}
  \usepackage[dvips]{graphicx}
  \usepackage[percent]{overpic}
  \usepackage[pdftex]{color}
  \usepackage{amsaddr}
  \usepackage[numbers,sort&compress]{natbib}

  \newcommand{\nequation}{\setcounter{equation}{0}}
  \renewcommand{\theequation}{\mbox{\arabic{section}.\arabic{equation}}}
  \newcommand{\R}{{\mathbb R}}
  
  \newcommand{\C}{{\mathbb C}}

  \newcommand{\re}{\text{\upshape Re\,}}
  
  \newcommand{\im}{\text{\upshape Im\,}}
  
  \newtheorem{theorem}{Theorem}[section]

  \newtheorem{lemma}[theorem]{Lemma}

  \newtheorem{remark}[theorem]{Remark}
  
  \newtheorem{figuretext}{Figure}

  \input epsf
  \title{Long-time asymptotics for the Nonlocal mKdV equation}

  \author{Fengjing He$^{1}$, Engui FAN$^{1}$ and Jian Xu$^{2}$}
  \address{\small $1$. School of Mathematical Sciences, Fudan University, Shanghai
  200433, P.R. China.\\
$2$. College of Science, University of Shanghai for Science and Technology, Shanghai 200093, P. R. China}
  \email{fjhe16@fudan.edu.cn; faneg@fudan.edu.cn; jianxu@usst.edu.cn}

\begin{document}
\baselineskip=18pt
  \begin{abstract}
  \baselineskip=18pt
  \noindent
  In this paper, we study the Cauchy problem with decaying initial data for the nonlocal modified
  Korteweg-de Vries equation (nonlocal mKdV)
  \[q_t(x,t)+q_{xxx}(x,t)-6q(x,t)q(-x,-t)q_x(x,t)=0,\]
 which can be viewed as
  a generalization of the local classical mKdV equation.  We first formulate the Riemann-Hilbert
  problem associated with  the Cauchy problem of the nonlocal mKdV equation.  Then we apply the Deift-Zhou nonlinear
  steepest-descent method to analyze  the long-time asymptotics for the
  solution of  the nonlocal mKdV equation. In contrast with the classical mKdV equation,
  we   find  some  new  and different results   on  long-time asymptotics for the nonlocal mKdV equation   and  some additional assumptions about the scattering data  are   made in our
  main results.
  \end{abstract}

  \maketitle



  \setcounter{tocdepth}{1}
  \tableofcontents

  \section{Introduction}\nequation
  The pioneering work for the nonlocal systems has been done by Ablowitz
  and Musslimani   when they studied the nonlocal NLS equation with
  PT symmetry \cite{AM1}. This research field has  attracted much  attention from both mathematics  and  the physical application of
  nonlinear optics and magnetics \cite{MECM, MMEC, GA}. Since the nonlocal NLS was found, a number of other nonlocal
   integrable systems has been introduced from a mathematical viewpoint. For instance,    from symmetry reduction
  of general AKNS system,  some new reverse space-time and reverse time nonlocal nonlinear integrable version
  of the NLS, mKdV, sine-Gordon equation were found \cite{AM3}.    Recently, Yang  constructed some new nonlocal
  integrable equations by simple variable transformations on local equations \cite{YY}.

   Like the local case, the
  nonlocal  integrable systems also possess integrable properties,   for example,  nonlocal NLS equation admits infinite number of conservation
  laws and can be solved by using the inverse scattering transform (IST) \cite{AM2}.  Some  exact solutions of   nonlocal mKdV equation  including soliton, kink,
  rogue-wave  and breather  were obtained through either Darboux transformation
  or IST.  These solutions have displayed  some new properties which are different from those
  of local equation\cite{JZ1, JZ2}.   In physical application, the nonlocal mKdV possesses the shifted
  parity and delayed time reveal symmetry, and thus it can be related to the
  Alice-Bob   system\cite{LH}.   For instance, a special approximate solution of the
  nonlocal mKdV was applied to theoretically capture the salient features of two
  correlated dipole blocking events in atmospheric dynamical systems \cite{TLH}. 
  
    However,  there 
  has been still  not  much work on the Riemann-Hilbert method to the nonlocal systems except 
  to the recent paper \cite{RS}, where  Rybalko and Shepelsky  obtained  the long-time asymptotics of the solution
 for the nonlocal Schrodinger equation via the nonlinear steepest-descent method.
In this paper,   we  apply Riemann-Hilbert (RH)  method and Deift-Zhou nonlinear steepest-descent method 
  to analyze  longtime asymptotics  of  the Cauchy problem of  the nonlocal mKdV equation
  \begin{subequations}\label{nmkdv_ini}
    \begin{align}
      \label{nmkdv}
      &q_t(x,t)+q_{xxx}(x,t)-6q(x,t)r(x,t)q_x(x,t)=0,\\
      \label{ini}
      &q(x,0)=q_0(x),
    \end{align}
  \end{subequations}
  where $r(x,t)=q(-x,-t)$ is a symmetry reduction of an AKNS system,
  and the initial data $q_0(x)$ decays rapidly to zero  as $x\to\pm\infty$.

 In 1970's, the solutions of the Cauchy problem for many integrable nonlinear
  wave equations was obtained by solving an associated RH problem
  on the complex plane \cite{AKNS}. More precisely, starting with initial data,
  the direct scattering transform gives rise to certain spectral functions whose time
  evolution is simple. Then the solution of the original Cauchy problem can be
  recovered via the IST characterized in terms of RH problem whose jump
  matrix depends on the given spectral functions.

 In 1993,  Deift and Zhou introduced the nonlinear steepest-descent method to analyze the
  asymptotics of the solutions of RH problems \cite{DZ}. It involves a series
  of counter deformation aiming to reduce the original RH problem to the one
  whose jump matrix is decaying fast (as $t\to\infty$) to the identity matrix
  everywhere except near some stationary phase points; and it is the contour
  near these points that determine the leading order of the long time
  asymptotics which can be obtained explicitly after rescaling the RH problem.
  This method has been used to study rigorously the long$-$time asymptotics of a wide variety of integral systems, such as the mKdV equation \cite{DZ}  and the non-focusing NLS equation \cite{diz},the sine-Gordon equation \cite{cvz}, the modified Schr$\ddot{o}$dinger equation \cite{kv1,kv2}, the KdV equation \cite{gt}, the Cammasa$-$Holm equation \cite{bkst}, Fokas-Lenells equation\cite{XF}, derivative NLS equation \cite{XFC},
  short pulse equation\cite{X,XF2}, Sine-Gordon equation \cite{HL},  Kundu-Eckhaus Equation\cite{ZXF}.

In \cite{DZ}, Deift and Zhou obtained the explicit leading order long-time
asymptotic behavior of the solution  to the classical mKdV equation
  \begin{subequations}\label{mkdv_ini}
    \begin{align}
      \label{mkdv}
      &q_t(x,t)+q_{xxx}(x,t)-6q^2(x,t)q(x,t)=0,\\
      \label{ini}
      &q(x,0)=q_0(x),
    \end{align}
  \end{subequations}
 using the nonlinear steepest descent method.  
Here we extend above results to give the asymptotic behavior of solution   of  nonlocal mKdV equation (\ref{nmkdv_ini}), 
 but it will be much different  from that on the classical mKdV equation (\ref{mkdv_ini}) in the following three aspects.
  
 (i)   For our  nonlocal mKdV equation,  the  jump matrix of the  RH problem involve  two  reflection coefficients $r_1(k)$ and $r_2(k)$,   but 
  there is only  one  reflection coefficient $r(k)$ for  the local mKdV equation,  which is  specified by
$ r_1(k)=r(k), \quad r_2(k)=\overline{r(\overline{k})}, \ |r(k)|<1$.

(ii)   In the analysis of the local equations,   the  great   difference
    from the nonlocal case is that $1-r_1(k)r_1(k)$ is complex-valued,
    which leads  to  $\im \nu(\zeta) \neq 0$. We will find below that $\im \nu(\zeta)$
    contributes to both the leading order and the error terms in
    the asymptotics for the nonlocal mKdV equation.   To obtain asymptotic behavior of solution   of  nonlocal mKdV equation, 
    we have used Slightly different method from  that in \cite{DZ}.

(iii)  At last,  in  contrast with the asymptotic of  local mKdV equation, 
we obtain the long time asymptotic    of  nonlocal mKdV equation 
as follows
    \begin{align}\nonumber
      q(x,t)=\frac{4\epsilon\re\beta(\zeta,t)}{\tau^{\frac{1}{2}-\im\nu(\zeta)}}+
      O(\epsilon\tau^{-\frac{1+\alpha}{2}+|\im\nu(\zeta)|+\im\nu(\zeta)}).
    \end{align}
 Note that   the decay rate of the leading term depends
    on $\zeta=\frac{x}{t}$ through $\im\nu(\zeta)$,  while
    $\im\nu(\zeta)=0$ for all $\zeta\in\mathcal{I}$ in the local mKdV equation.

  Organization of this paper is as follows.
  In Section \ref{ist_rhp}, we present the IST and express the solution of nonlocal mKdV equation (\ref{nmkdv_ini})
  in terms of a RH problem. In Section \ref{rtamrp}, we conduct several deformations to
  obtain a model RH problem convenient for consequent analysis. In Section \ref{lta}, we derive the
  long-time behavior of nonlocal mKdV equation (\ref{nmkdv_ini}) in the similarity sector.

  \section{Inverse scattering transform and the Riemann-Hilbert problem}\nequation\label{ist_rhp}
  Since (\ref{nmkdv}) is a member of AKNS systems, the standard method of IST was applied
  in \cite{JZ2}. We reformulate the IST to express the solution of (\ref{nmkdv_ini})
  in terms of a RH problem    for   convenience   of  the consequent analysis.

  The nonlocal mKdV equation (\ref{nmkdv}) admits the Lax pair
  \begin{subequations}\label{lax}
    \begin{align}
      \label{lax_a}
      &\Phi_x+ik\sigma_3\Phi=U\Phi,\\
      \label{lax_b}
      &\Phi_t+4ik^3\sigma_3\Phi=V\Phi,
    \end{align}
  \end{subequations}
  where $\Phi(x,t,k)$ is a $2\times2$-matrix valued eigenfunction, $k\in\C$ is the
  spectral parameter, and
  \[\sigma_3=\begin{pmatrix}
    1& 0\\
    0& -1\\
  \end{pmatrix}, \qquad
  U=\begin{pmatrix}
    0& q(x,t)\\
    q(-x,-t)& 0\\
  \end{pmatrix}, \qquad
  V=\begin{pmatrix}
    A& B\\
    C& -A\\
  \end{pmatrix}\]
  with
  \begin{align*}
    &A=-2iq(x,t)q(-x,-t)k+q(-x,-t)q_x(x,t)+q(x,t)q_x(-x,-t),\\
    &B=4k^2q(x,t)+2iq_x(x,t)k+2q^2(x,t)q(-x,-t)-q_{xx}(x,t),\\
    &C=4k^2q(-x,-t)+2iq_x(-x,-t)k+2q^2(x,t)q(-x,-t)-q_{xx}(-x,-t).
  \end{align*}

  Let $\Psi_j(x,t,k)$, $j=1,2$, be the $2\times2$-matrix valued solutions of the linear
  Volterra integral equations
  \begin{subequations}\label{volterra}
    \begin{align}
      \label{volterra_1}
      &\Psi_1(x,t,k)=I+\int_{-\infty}^{x}e^{ik(y-x)\hat{\sigma}_3}(U(y,t)\Psi_1(y,t,k))\,dy,\quad
      k\in(\C_+,\C_-), \\
      \label{volterra_2}
      &\Psi_2(x,t,k)=I+\int_{\infty}^{x}e^{ik(y-x)\hat{\sigma}_3}(U(y,t)\Psi_2(y,t,k))\,dy,\quad
      k\in(\C_-,\C_+),
    \end{align}
  \end{subequations}
  where $\hat{\sigma}_3$ acts on a $2\times2$ matrix $A$ by
  $\hat{\sigma}_3A=[\hat{\sigma}_3,A]$, i.e.
  $e^{\hat{\sigma}_3}A=e^{\sigma_3}Ae^{-\sigma_3}$,
  $\C_{\pm}=\left\{k\in\C\,|\pm\im k>0\right\}$, and the notation
  $k\in(\C_+,\C_-)$ indicates that the first and second columns are valid for
  $k\in\C_+$ and $k\in\C_-$, respectively. From (\ref{volterra}), we can prove
  that $\Psi_1(x,t,\cdot)$ is continuous for
  $k\in(\overline{\C_+},\overline{\C_-})$ and analytic for $k\in(\C_+,\C_-)$,
  $\Psi_2(x,t,\cdot)$ is continuous for $k\in(\overline{\C_-},\overline{\C_+})$
  and analytic for $k\in(\C_-,\C_+)$ \cite{BC}. Moreover we can derive the large
  $k$ asymptotics of $\Psi_j$ (c.f. \cite{BC})
  \begin{equation}\label{kasym}
    \Psi_j(x,t,k)=I+O(k^{-1}), \qquad k\to\infty,
  \end{equation}
  where the error term is uniformly with respect to $x$, $t$.

  Then the Jost solutions
  $\Phi_j(x,t,k)$, $j=1,2$, of (\ref{lax}) are defined as follow
  \begin{equation}\label{jost}
    \Phi_j(x,t,k)=\Psi_j(x,t,k)e^{(-ikx-4ik^3t)\sigma_3}.
  \end{equation}
  Since $U$ is traceless, $\det\Phi_j(x,t,k)\equiv1$ for all $x$, $t$, and $k$.
  And for $k\in\R$, $\Phi_j(x,t,k)$ can be related by scattering matrix $S(k)$
  \begin{equation}\label{jost1_2}
    \Phi_1(x,t,k)=\Phi_2(x,t,k)S(k),\qquad k\in\R,
  \end{equation}
  where
  \begin{equation}
    \label{scat}
    S(k)=\begin{pmatrix}
    s_{11}(k)& s_{12}(k)\\
    s_{21}(k)& s_{22}(k)\\
    \end{pmatrix}
    ,\qquad k\in\R,
  \end{equation}
  is independent of $x$ and $t$.

  We now establish important symmetry properties of the scattering matrix (\ref{scat}).
  It can be verified that if $\Psi(x,t,k)$ is the solution of (\ref{volterra_1}), then
  $\Lambda\overline{\Psi(-x,-t,-\bar{k})}\Lambda^{-1}$ is the solution of
  (\ref{volterra_2}) with
  $\Lambda=\bigl(
    \begin{smallmatrix}
      0 &-1\\1 &0
    \end{smallmatrix})$.
  Notice (\ref{jost}) and the uniqueness of the solution of the Volterra equation
  (\ref{volterra}), we arrives that
  \begin{equation}\label{symm}
    \Phi_2^{(2)}(x,t,k)=\Lambda\overline{\Phi_1^{(1)}(-x,-t,-\overline{k})}, \qquad
    \Phi_2^{(1)}(x,t,k)=\Lambda^{-1}\overline{\Phi_1^{(2)}(-x,-t,-\overline{k})},
  \end{equation}
  where $\Phi_i^{(j)}(x,t,k)$ denotes the $j$-th column of the matrix $\Phi_i(x,t,k)$.
  Rewrite the relation between the Jost solutions (\ref{jost1_2}) as
  \begin{subequations}\label{jost_re}
    \begin{align}
      &\Phi_1^{(1)}(x,t,k)=s_{11}(k)\Phi_2^{(1)}(x,t,k)+s_{21}(k)\Phi_2^{(2)}(x,t,k),\\
      &\Phi_1^{(2)}(x,t,k)=s_{12}(k)\Phi_2^{(1)}(x,t,k)+s_{22}(k)\Phi_2^{(2)}(x,t,k),
    \end{align}
  \end{subequations}
  the scattering data can be represented in terms of $\Phi_i^{(j)}$,
  and from (\ref{symm}), we reach the following symmetry
  \begin{subequations}\label{scat_symm}
    \begin{align}\nonumber
      s_{11}(k)&=\det(\Phi_1^{(1)}(x,t,k),\Phi_2^{(2)}(x,t,k))\\
      &=\det(\Lambda^{-1}(\overline{\Phi_1^{(1)}(-x,-t,-\overline{k})},\overline{\Phi_2^{(2)}(-x,-t,-\overline{k})})\Lambda)
      =\overline{s_{11}(-\overline{k})},\\
      \nonumber
      s_{22}(k)&=\det(\Phi_2^{(1)}(x,t,k),\Phi_1^{(2)}(x,t,k))\\
      &=\det(\Lambda^{-1}(\overline{\Phi_2^{(1)}(-x,-t,-\overline{k})},\overline{\Phi_1^{(2)}(-x,-t,-\overline{k})})\Lambda)
      =\overline{s_{22}(-\overline{k})},\\
      \nonumber
      s_{12}(k)&=\det(\Phi_1^{(2)}(x,t,k),\Phi_2^{(2)}(x,t,k))\\
      &=\det(\Lambda(\overline{\Phi_2^{(1)}(-x,-t,-\overline{k})},\overline{\Phi_1^{(1)}(-x,-t,-\overline{k})}))
      =\overline{s_{21}(-\overline{k})}.
    \end{align}
  \end{subequations}
  Further more, we can also verify that if $\Psi(x,t,k)$ is the solution of
  (\ref{volterra_1}), then $\Lambda \Psi(-x,-t,k) \Lambda^{-1}$
  is the solution of (\ref{volterra_2}). So following the above procedure, we obtain
  another symmetry property
  \begin{equation}\label{ano_scat_symm}
    s_{12}(k)=s_{21}(k).
  \end{equation}
  Finally, from (\ref{scat_symm}) and (\ref{ano_scat_symm}), $S(k)$ can be written
  in the form
  \begin{equation}\label{S}
    S(k)=\begin{pmatrix}
      a_1(k)& b(k)\\
      b(k)& a_2(k)
    \end{pmatrix},
  \end{equation}
  where
  \begin{equation}\label{absymm}
    a_1(k)=\overline{a_1(-\overline{k})}, \qquad a_2(k)=\overline{a_2(-\overline{k})},
    \qquad b(k)=\overline{b(-\overline{k})}.
  \end{equation}

  In accordance with the case of local equations \cite{BC, IN}, the scattering matrix
  $S(k)$ is uniquely determined by the initial data $q_0(x)$, and we can conclude that:
  \begin{enumerate}
    \item
    $a_1(k)$ is analytic for $k\in\C_+$, and continuous for $k\in\overline{\C_+}$;
    $a_2(k)$ is analytic for $k\in\C_-$, and continuous for $k\in\overline{\C_-}$.
    \item
    $a_j(k)=1+O(k^{-1})$, $j=1,2$ and
    $b(k)=O(k^{-1})$ as $k\to\infty$.
    \item
    $a_1(k)=\overline{a_1(-\overline{k})},\, k\in\overline{\C_+};\qquad
    a_2(k)=\overline{a_2(-\overline{k})},\, k\in\overline{\C_-};\qquad
    b(k)=\overline{b(-\overline{k})},\, k\in\R$.
    \item\label{determinant}
    $a_1(k)a_2(k)-b(k)^2=1$, $k\in{\R}$, (follows from $\det S(k)=1$).
  \end{enumerate}

  Now we define the matrix valued function $M$ as
  \begin{equation}\label{M}
    M(x,t,k)=\begin{cases}
    \left(\frac{\Psi_1^{(1)}(x,t,k)}{a_{1}(k)},\Psi_2^{(2)}(x,t,k)\right), \im k>0, \\
    \left(\Psi_2^{(1)}(x,t,k),\frac{\Psi_1^{(2)}(x,t,k)}{a_{2}(k)}\right), \im k<0.
    \end{cases}
  \end{equation}
  Using scattering relation (\ref{jost1_2}), we have the jump condition for
  $M(x,t,k)$ across $k\in\R$
  \begin{equation}
    \label{MMJ}
    M_{+}(x,t,k)=M_{-}(x,t,k)J(x,t,k),\qquad k \in \R,
  \end{equation}
  where $M_{\pm}$ is the limiting value of $M$ as $k$ approaches $\R$ from $\C_\pm$,
  and
  \begin{equation}
    \label{J}
    J(x,t,k)=e^{(-ikx-4ik^3t)\hat{\sigma}_3}
    \begin{pmatrix}
      1-r_1(k)r_2(k)& -r_2(k)\\
      r_1(k)& 1
    \end{pmatrix},\qquad k \in \R,
  \end{equation}
  and reflection coefficients are defined by
  \begin{equation}\label{r1r2}
    r_1(k)=\frac{b(k)}{a_1(k)},\qquad
    r_2(k)=\frac{b(k)}{a_2(k)}.
  \end{equation}
  From the symmetry of scattering data (\ref{absymm}), $r_1$ and $r_2$ also possess
  the symmetry property
  \begin{equation}\label{r1r2symm}
    r_1(k)=\overline{r_1(-\overline{k})}, \qquad
    r_2(k)=\overline{r_2(-\overline{k})}, \qquad k \in \R,
  \end{equation}
  and the determinant property 4 implies that
  \begin{equation}
    1-r_1(k)r_2(k)=\frac{1}{a_1(k)a_2(k)}, \qquad k \in \R.
  \end{equation}

  We assume that $a_1(k)$ and $a_2(k)$ have no zeros in $\overline{\C_+}$ and
  $\overline{\C_-}$ respectively so that one can assemble the above facts into
  the form of a Riemann-Hilbert problem
  \begin{equation}\label{Mrhp}
    \begin{cases}
      M(x,t,k) \text{ analytic in } \C\setminus\R,\\
      M_{+}(x,t,k)=M_{-}(x,t,k)J(x,t,k),\qquad k \in \R,\\
      M(x,t,k)\to I, \qquad k\to\infty.
    \end{cases}
  \end{equation}

  \begin{remark}\label{r1r2diff}
    RH problem (\ref{Mrhp}) can be regard as a generalization of
    the RH problem associated with the mKdV equation. In the local
    case, the reflection coefficients are specified by
    \begin{equation}\label{localr}
      r_1(k)=r(k), \quad r_2(k)=\overline{r(\overline{k})}, \qquad k \in \R,
    \end{equation}
    with $|r(k)|<1$.
  \end{remark}

  Inversely, if RH problem (\ref{Mrhp}) has a unique solution for all $(x,t)$,
  the solution $q(x,t)$ of (\ref{nmkdv_ini}) is given by (c.f. \cite{AC, BDT, FT})
  \begin{equation}\label{ist4.35}
    q(x,t)=2i\lim_{k\to\infty}(kM(x,t,k))_{12}.
  \end{equation}

  \section{Reduction to a model RH problem}\nequation\label{rtamrp}
  The deformations of the RH problem (\ref{Mrhp}) are similar to the local
  case \cite{DZ,L}, where the original RH problem is deformed to the one
  whose jump matrix decays to the identity matrix everywhere
  as $t \to \infty$ except near the stationary points. Then an explicitly
  solvable RH problem is introduced to obtain a model RH problem
  for which long time asymptotics can be conveniently performed.

  Let $\mathcal{I}=[-N,0)$ be the interval with $N>0$
  and let $\zeta=x/t$ be the variable with $\zeta \in \mathcal{I}$.
  Let $M(x,t,\cdot)$ denote the unique solution of the RH problem (\ref{Mrhp}),
  and the phase of the exponentials $e^{\pm t\Phi(\zeta)}$ in the jump matrix
  (\ref{J}) is defined by
  \begin{equation}\label{phasePhi}
    \Phi(\zeta,k)=2ik\zeta+8ik^3,
  \end{equation}
which admits  two stationary points
  \begin{equation}\label{+-k_0}
    k_0=\sqrt{-\frac{\zeta}{12}}, \qquad k_0=-\sqrt{-\frac{\zeta}{12}},
  \end{equation}
  such that $\frac{d\Phi}{dk}(\pm k_0)=0$.

  Now we deform the RH problem (\ref{Mrhp})
  with the following steps.

  {\bf Step 1}
  The first step is to search for  upper/lower and lower/upper triangular factorizations  of the jump matrix.
  For this purpose, we  introduce  a scalar RH problem
  \begin{equation}\label{deltarhp}
    \begin{cases}
      \delta \text{ analytic in } \C \setminus [-k_0,k_0],\\
      \delta_+=\delta_-(1-r_1(k)r_2(k)), \qquad |k|<k_0,\\
      \delta \to 1 \qquad k \to \infty.
    \end{cases}
  \end{equation}
  Direct calculation shows that (\ref{deltarhp}) admits a unique solution
  \begin{equation}\label{delta}
    \delta(\zeta,k)=e^{\frac{1}{2\pi i} \int_{-k_0}^{k_0} \ln(1- r_1(s)r_2(s)) \frac{ds}{s - k}}, \qquad  k \in \C \setminus [-k_0,k_0].
  \end{equation}
  The symmetry (\ref{r1r2symm}) implies that
  \begin{equation}\label{deltasymm}
    \delta(\zeta,k)=\overline{\delta(\zeta,-\overline{k})},
  \end{equation}
  moreover, integrating by parts in formula (\ref{delta}) yields
  \begin{equation}\label{deltanutildechi}
    \delta(\zeta,k)=\frac{(k-k_0)^{i\nu(\zeta)}}{(k+k_0)^{i\overline{\nu(\zeta)}}}e^{\tilde{\chi}(\zeta,k)},
  \end{equation}
  where $\tilde{\chi}(\zeta,k)$ is a uniformly bounded function
  with respect to $\zeta \in \mathcal{I}$ and $k \in \C \setminus \R$,
  which is defined by
  \begin{equation}\label{tildechi}
    \tilde{\chi}(\zeta,k)=-\frac{1}{2\pi i} \int_{-k_0}^{k_0} \ln(k-s) d\ln(1-r_1(s)r_2(s)),
  \end{equation}
  and $\nu(\zeta)$ is a bounded function defined by
  \begin{equation}\label{nu}
    \nu(\zeta)=-\frac{1}{2\pi}\ln(1-r_1(k_0)r_2(k_0))=-\frac{1}{2\pi}\ln|1-r_1(k_0)r_2(k_0)|-\frac{i}{2\pi}\Delta(\zeta)
  \end{equation}
  with
  \begin{equation*}
    \Delta(\zeta)=\int_{-\infty}^{k_0} d\arg(1-r_1(s)r_2(s)).
  \end{equation*}

  We assume that
  \begin{equation}\label{Deltares}
    \Delta(\zeta) \in (-\pi, \pi), \qquad \zeta\in\mathcal{I},
  \end{equation}
  then $\nu(\zeta)$ is single valued and
  \begin{equation}\label{reinu}
    \bigl|\im \nu(\zeta)\bigr| < \frac{1}{2}, \qquad \zeta \in \mathcal{I}.
  \end{equation}
  Consequently the singularity of $\delta(\zeta,k)$ at $k=\pm k_0$ is square
  integrable.

  $\delta(\zeta,k)$ can be written in another way:
  \begin{equation}\label{deltachi}
    \delta(\zeta,k)=\bigl(\frac{k-k_0}{k+k_0}\bigr)^{i\nu(\zeta)}e^{\chi(\zeta,k)},
  \end{equation}
  where
  \begin{align}\label{chi}\nonumber
    \chi(\zeta,k)&=\frac{1}{2\pi i}\int_{-k_0}^{k_0}
    \ln\bigl(\frac{1-r_1(s)r_2(s)}{1-r_1(k_0)r_2(k_0)}\bigr)
    \frac{ds}{s-k}\\
    &=\tilde{\chi}(\zeta,k)-\frac{1}{2\pi i}
    \ln \bigl(\frac{1-\overline{r_1(k_0)}\overline{r_2(k_0)}}{1-r_1(k_0)r_2(k_0)}\bigr)
    \ln(k+k_0).
  \end{align}
  In the local case, $\chi(\zeta,k)$ is equivalent to
  $\tilde{\chi}(\zeta,k)$ by symmetry (\ref{localr}), so $\chi(\zeta,k)$
  is uniformly bounded. However $\chi(\zeta,k)$ is singular at
  $k=-k_0$ for nonlocal equation.

  \begin{lemma}\label{chiuniformlybounded}
    Let $S=\bigl\{k' \in \C \big| |k'+k_0| \geq \frac{k_0}{2}\bigr\}$
    denote the complex plane minusing a neighborhood of $-k_0$.
    Then $\chi(\zeta, k)$ is uniformly bounded with respect to
    $\zeta \in \mathcal{I}$ and
    $k \in S$, i.e.
    \begin{equation}\label{supsupchileq}
      \sup_{\zeta \in \mathcal{I}}
      \sup_{k \in S}
      |\chi(\zeta,k)| \leq C
    \end{equation}
  \end{lemma}

  \begin{proof}
    Since $\delta(\zeta, k) \to 1$ as $k \to \infty$, $\chi(\zeta,k)$ is uniformly
    bounded with respect to $\zeta \in \mathcal{I}$ and
    $k \in \{k' \in \C | |k'+k_0| \geq G\}$ by (\ref{deltachi}),
    where $G$ is a large enough constant.
    Let $k=-k_0+ue^{i\alpha}$, where $\frac{k_0}{2} \leq u < G$ and
    $\alpha \in (-\pi, \pi]$. By (\ref{chi})
    \begin{align}\label{chileq}\nonumber
      |\chi(\zeta,k)| & \leq C + C\Biggl|\ln\bigg(\frac{1-\overline{r_1(k_0)}\overline{r_2(k_0)}}{1-r_1(k_0)r_2(k_0)}\bigg)\Biggr| \big|\ln(ue^{i\alpha})\big|\\
      & \leq C + C\Biggl|\ln\bigg(\frac{1-\overline{r_1(k_0)}\overline{r_2(k_0)}}{1-r_1(k_0)r_2(k_0)}\bigg)\Biggr| \biggl|\ln\frac{k_0}{2}\biggr|.
    \end{align}
    Symmetry (\ref{r1r2symm}) implies that
    \begin{align}\label{r10r20}
      r_j(0)=\overline{r_j(0)}, \qquad j=1,2,
    \end{align}
    thus $\chi(\zeta, k)$ is also uniformly bounded with respect
    to $\zeta \in \mathcal{I}$ and
    $k \in \bigl\{k' \in \C \big| \frac{k_0}{2}\leq |k'+k_0| < G \bigr\}$
    by (\ref{chileq}) and (\ref{r10r20})
  \end{proof}

  \begin{remark}\label{nmkdvdiff}
    In the analysis of the local equations, the chief difference
    from the nonlocal case is that $1-r_1(k)r_2(k)$ is complex-valued,
    that is, $\im \nu(\zeta) \neq 0$. We will find below that $\im \nu(\zeta)$
    contributes to both the leading order and the error terms in
    the asymptotics for the nonlocal mKdV equation.%, and
    % a further restriction to the range of $\Delta(k)$ is required
    % to ensure that the error terms are approriately estimated.
  \end{remark}

  Conjugating the RH problem (\ref{Mrhp}) by
  \begin{equation}\label{deltasig3}
    \delta(\zeta,k)^{-\sigma_3}=
    \begin{pmatrix}
      \delta(\zeta,k)^{-1} & 0\\
      0 & \delta(\zeta,k)
    \end{pmatrix}
  \end{equation}
  leads to the factorization problem for $\tilde{M}(x,t,k)=M(x,t,k)\delta(\zeta,k)^{-\sigma_3}$,
  \begin{equation}\label{tildeMrhp}
    \begin{cases}
      \tilde{M}_+(x,t,k)=\tilde{M}_-(x,t,k) \tilde{J}(x,t,k), \qquad k \in \R,\\
      \tilde{M}(x,t,k) \to I, \qquad k \to \infty,
    \end{cases}
  \end{equation}
  where
  \begin{align}\label{tildeJ}
    \tilde{J}=
    \begin{cases}
      \begin{pmatrix}
        1 & -\delta(\zeta,k)^2r_2(k)e^{-t\Phi(\zeta,k)}\\
        0 & 1
      \end{pmatrix}
      \begin{pmatrix}
        1 & 0\\
        \delta(\zeta,k)^{-2}r_1(k)e^{t\Phi(\zeta,k)} & 1
      \end{pmatrix}, \,\, &|k|>k_0,\\
      \begin{pmatrix}
        1 & 0\\
        \delta_-(\zeta,k)^{-2}r_3(k)e^{t\Phi(\zeta,k)} & 1
      \end{pmatrix}
      \begin{pmatrix}
        1 & -\delta_+(\zeta,k)^2r_4(k)e^{-t\Phi(\zeta,k)}\\
        0 & 1
      \end{pmatrix}, \,\, &|k|<k_0,
    \end{cases} \qquad k \in \R,
  \end{align}
  with
  \begin{equation}\label{r3r4}
    r_3(k)=\frac{r_1(k)}{1-r_1(k)r_2(k)}, \qquad
    r_4(k)=\frac{r_2(k)}{1-r_1(k)r_2(k)}.
  \end{equation}

  {\bf Step 2}
  In accordance with the local case,
  we introduce oriented counter $\Gamma$ and open sets $\{V_j\}_1^6$
  as depicted in Figure \ref{Gamma.pdf},
  \begin{figure}
    \begin{center}
      \begin{overpic}[width=.7\textwidth]{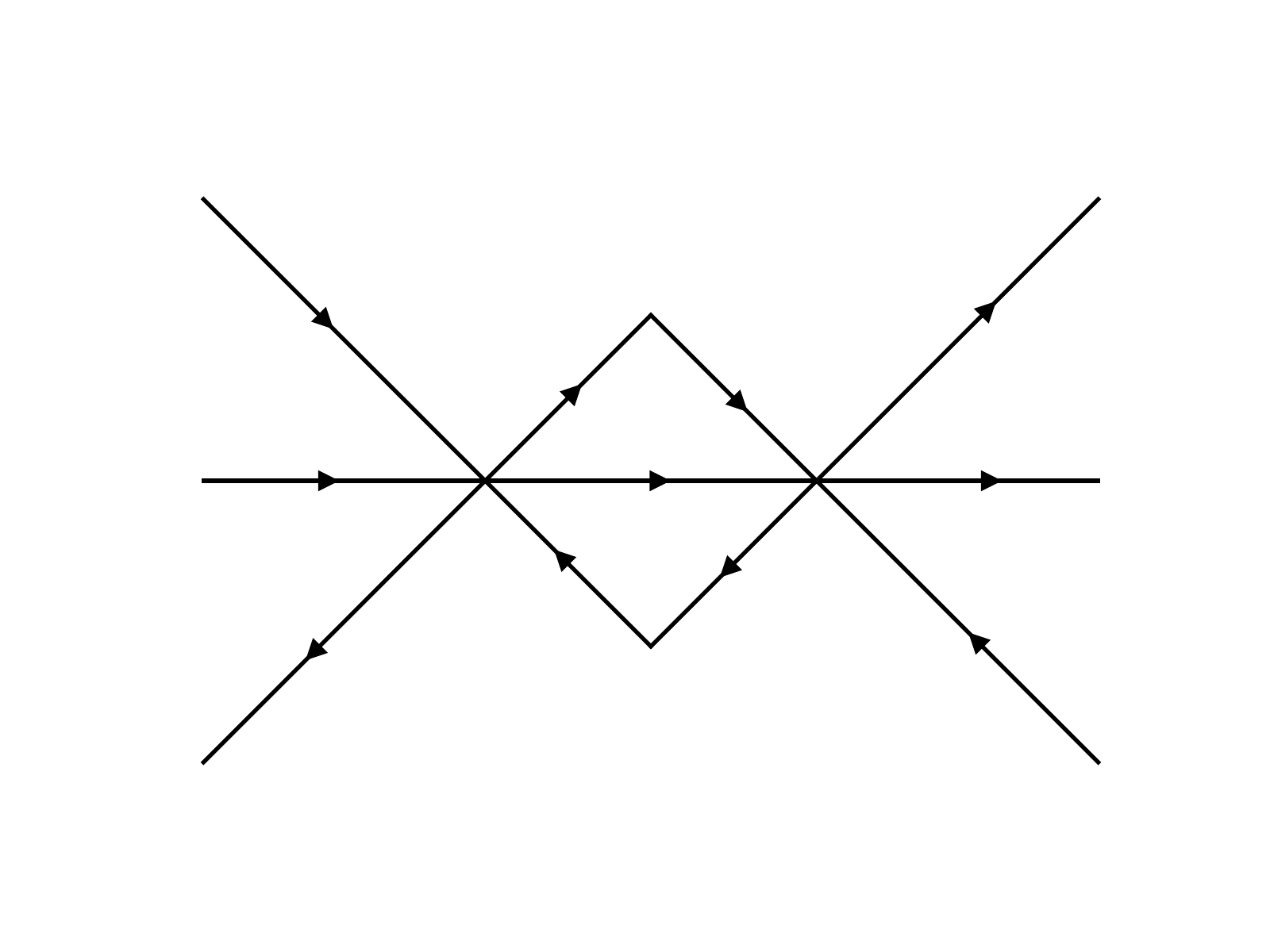}
        \put(23,41){$V_1$}
        \put(75,41){$V_1$}
        \put(50,58){$V_2$}
        \put(50,40){$V_3$}
        \put(50,30){$V_4$}
        \put(50,14){$V_5$}
        \put(23,30){$V_6$}
        \put(75,30){$V_6$}
        \put(35,31){$-k_0$}
        \put(63,31){$k_0$}
        \put(90,36){$\Gamma$}
      \end{overpic}
      \begin{figuretext}\label{Gamma.pdf}
        The jump contour $\Gamma$ and the open sets $\{V_j\}_1^6$
      \end{figuretext}
    \end{center}
  \end{figure}
  and define $m(x,t,k)$ by
  \begin{align}\label{m}
    m(x,t,k)=
    \begin{cases}
      \tilde{M}(x,t,k)
      \begin{pmatrix}
        1 & 0\\
        -\delta(\zeta,k)^{-2}r_{1,a}(x,t,k)e^{t\Phi(\zeta,k)} & 1
      \end{pmatrix},\qquad &k \in V_1,\\
      \tilde{M}(x,t,k)
      \begin{pmatrix}
        1 & \delta(\zeta,k)^2r_{4,a}(x,t,k)e^{-t\Phi(\zeta,k)}\\
        0 & 1
      \end{pmatrix},\qquad &k \in V_3,\\
      \tilde{M}(x,t,k)
      \begin{pmatrix}
        1 & 0\\
        \delta(\zeta,k)^{-2}r_{3,a}(x,t,k)e^{t\Phi(\zeta,k)} & 1
      \end{pmatrix},\qquad &k \in V_4,\\
      \tilde{M}(x,t,k)
      \begin{pmatrix}
        1 & -\delta(\zeta,k)^2r_{2,a}(x,t,k)e^{-t\Phi(\zeta,k)}\\
        0 & 1
      \end{pmatrix},\qquad &k \in V_6,\\
      \tilde{M}(x,t,k),\qquad &elsewhere,
    \end{cases}
  \end{align}
  where $r_{j,a}$ is the analytic approximation of $r_j$
  with small error $r_{j,r},\,j=1,\cdots,4$.
  More precisely, since $\{r_j(k)\}_1^4$ are sufficiently smooth and
  decaying, we can closely follow the proof of Lemma 4.8 in \cite{L}
  to obtain similar decompositions:

  Dividing the complex $k$-plane into four parts $U_j,\, j=1,\cdots,4$
  as in Figure \ref{U.pdf} so that
  \begin{equation}\label{U1234}
    \{k|\re \Phi(\zeta,k)<0\}=U_1 \cup U_3, \qquad
    \{k|\re \Phi(\zeta,k)>0\}=U_2 \cup U_4,
  \end{equation}
  \begin{figure}
    \begin{center}
      \begin{overpic}[width=.7\textwidth]{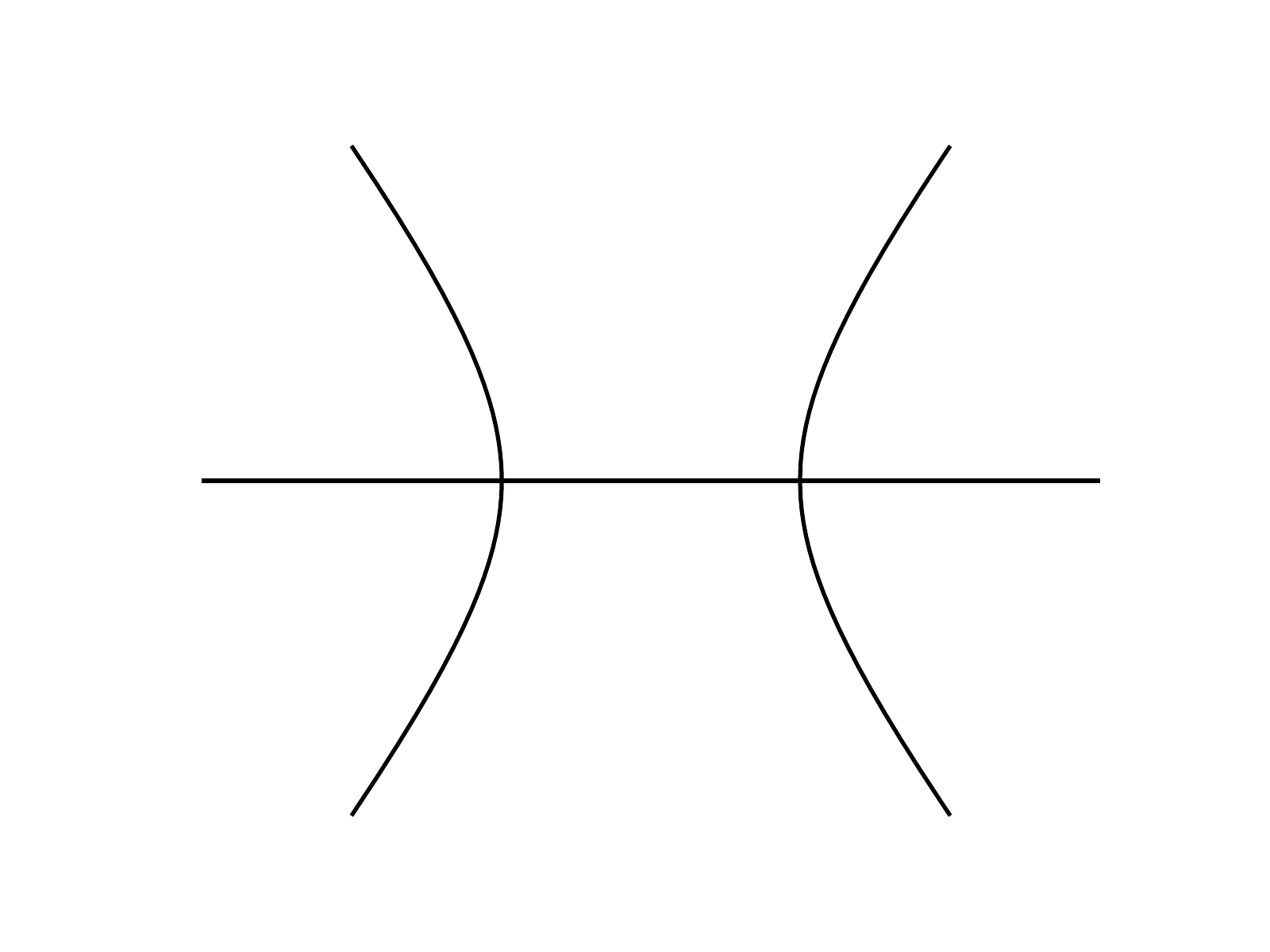}
        \put(25,46){$U_1$}
        \put(72.5,46){$U_1$}
        \put(25,24){$U_2$}
        \put(72.5,24){$U_2$}
        \put(49.5,24){$U_3$}
        \put(49.5,46){$U_4$}
        \put(29,32){$-k_0$}
        \put(65,32){$k_0$}
        \end{overpic}
      \begin{figuretext}\label{U.pdf}
        The domains $\{U_j\}_1^4$ in the complex $k$-plane. $\re \Phi = 0$ on the curves.
      \end{figuretext}
    \end{center}
  \end{figure}
  we can introduce decompositions
  \begin{equation}\label{rjrjarjr}
    r_j(k)=
    \begin{cases}
      r_{j,a}(x,t,k)+r_{j,r}(x,t,k), \qquad j=1,2, \, |k|>k_0, k \in \R,\\
      r_{j,a}(x,t,k)+r_{j,r}(x,t,k), \qquad j=3,4, \, |k|<k_0, k \in \R,
    \end{cases}
  \end{equation}
  such that
  \begin{enumerate}
    \item $r_{j,a}(x,t,k)$ is defined and continuous for $k \in \overline{U}_j$,
    analytic for $k \in U_j$, and for each $K>0$ satisfies
    \begin{align}\nonumber
      |r_{j,a}(x,t,k)-r_j(k_0)| \leq &C_K|k-k_0|e^{\frac{t}{4}|\re \Phi(\zeta,k)|},\\
      \label{rja-rjk0}
      &k \in \overline{U}_j, \quad |k| \leq K, \quad \zeta \in \mathcal{I}, \quad t>0, \quad j=1,\cdots,4,
    \end{align}
    where the constant $C$ is independent of $\zeta, t, k$.
    \item $r_{1,a}$ and $r_{2,a}$ satisfy
    \begin{equation}\label{r1ar2a}
      |r_{j,a}(x,t,k)| \leq \frac{C}{1+|k|}e^{\frac{t}{4}|\re \Phi(\zeta,k)|},
      \qquad k \in \overline{U}_j, \quad \zeta \in \mathcal{I}, \quad t>0, \quad j=1,2,
    \end{equation}
    where the constant $C$ is independent of $\zeta, t, k$.
    % \item The $L^1,L^2,$ and $L^{\infty}$ norms on
    % $(-\infty,-k_0)\cup(k_0,\infty)$ of $r_{1,r}(x,t,\cdot)$ and
    % $r_{2,r}(x,t,\cdot)$ are $O(t^{-3/2})$ as $t \to \infty$ uniformly
    % with respect to $\zeta \in \mathcal{I}$.
    % \item The $L^1,L^2,$ and $L^{\infty}$ norms on
    % $(-k_0,k_0)$ of $r_{3,r}(x,t,\cdot)$ and
    % $r_{4,r}(x,t,\cdot)$ are $O(t^{-3/2})$ as $t \to \infty$ uniformly
    % with respect to $\zeta \in \mathcal{I}$.
    \item $r_{1,r}$ and $r_{2,r}$ satisfy
    \begin{align}\nonumber
      |r_{j,r}(x,t,k)|\leq C\frac{|k-k_0|}{1+|k|^2}t^{-3/2},\qquad
      &k\in(-\infty, -k_0) \cup (k_0, \infty),\\\label{r1rr2r}
      &\zeta\in\mathcal{I}, \quad t>0, \quad j=1,2,
    \end{align}
    where the constant $C$ is independent of $\zeta, t, k$.
    % \item On $(-\infty, -k_0) \cup (k_0, \infty)$, the
    % $L^1,L^2,$ and $L^{\infty}$ norms of $r_{1,r}$ and $r_{2,r}$
    % satisfy
    % \begin{equation}\label{r1rr2r}
    %   \Vert r_{j,r}(x,t,\cdot) \Vert_{L^p}=O(t^{-3/2}), \qquad j=1,2, \quad p=1,2,\infty,
    % \end{equation}
    % as $t \to \infty$ uniformly with respect to $\zeta \in \mathcal{I}$
    \item $r_{3,r}$ and $r_{4,r}$ satisfy
    \begin{align}\nonumber
      |r_{j,r}(x,t,k)|\leq C|k^2-k_0^2|t^{-3/2},\qquad
      &k\in(-k_0, k_0),\\\label{r3rr4r}
      &\zeta\in\mathcal{I}, \quad t>0, \quad j=1,2,
    \end{align}
    where the constant $C$ is independent of $\zeta, t, k$.
    % \item On $(-k_0, k_0)$, the $L^1,L^2,$ and $L^{\infty}$
    % norms of $r_{3,r}$ and $r_{4,r}$ satisfy
    % \begin{equation}\label{r3rr4r}
    %   \Vert r_{j,r}(x,t,\cdot) \Vert_{L^p}=O(t^{-3/2}), \qquad j=3,4, \quad p=1,2,\infty,
    % \end{equation}
    % as $t \to \infty$ uniformly with respect to $\zeta \in \mathcal{I}$
    \item The following symmetries are valid:
    \begin{equation}\label{rjarjrsymm}
      r_{j,a}(\zeta,t,k)=\overline{r_{j,a}(\zeta,t,-\overline{k})}, \,\,
      r_{j,r}(\zeta,t,k)=\overline{r_{j,r}(\zeta,t,-\overline{k})}, \qquad
      j=1,\cdots,4.
    \end{equation}
  \end{enumerate}

  As a result, the function $m(x,t,k)$ satisfies the RH problem
  \begin{equation}\label{mrhp}
    \begin{cases}
      m_+(x,t,k)=m_-(x,t,k)v(x,t,k), \qquad k\in\Gamma,\\
      m(x,t,k) \to I, \qquad k\to\infty,
    \end{cases}
  \end{equation}
  where the jump matrix is
  \begin{align}\label{v}
    v(x,t,k)=
    \begin{cases}
      \begin{pmatrix}
        1 & 0\\
        \delta(\zeta,k)^{-2}r_{1,a}(x,t,k)e^{t\Phi(\zeta,k)} & 1
      \end{pmatrix},\quad &k \in \overline{V}_1\cap\overline{V}_2,\\
      \begin{pmatrix}
        1 & -\delta(\zeta,k)^2r_{4,a}(x,t,k)e^{-t\Phi(\zeta,k)}\\
        0 & 1
      \end{pmatrix},\quad &k \in \overline{V}_2\cap\overline{V}_3,\\
      \begin{pmatrix}
        1 & 0\\
        -\delta(\zeta,k)^{-2}r_{3,a}(x,t,k)e^{t\Phi(\zeta,k)} & 1
      \end{pmatrix},\quad &k \in \overline{V}_4\cap\overline{V}_5,\\
      \begin{pmatrix}
        1 & \delta(\zeta,k)^2r_{2,a}(x,t,k)e^{-t\Phi(\zeta,k)}\\
        0 & 1
      \end{pmatrix},\quad &k \in \overline{V}_5\cap\overline{V}_6,\\
      \begin{pmatrix}
        1-r_{1,r}(x,t,k)r_{2,r}(x,t,k) & -\delta(\zeta,k)^2r_{2,r}(x,t,k)e^{-t\Phi(\zeta,k)}\\
        \delta(\zeta,k)^{-2}r_{1,r}(x,t,k)e^{t\Phi(\zeta,k)} & 1
      \end{pmatrix},\quad &k \in \overline{V}_1\cap\overline{V}_6,\\
      \begin{pmatrix}
        1 & -\delta(\zeta,k)^2r_{4,r}(x,t,k)e^{-t\Phi(\zeta,k)}\\
        \delta(\zeta,k)^{-2}r_{3,r}(x,t,k)e^{t\Phi(\zeta,k)} & 1-r_{3,r}(x,t,k)r_{4,r}(x,t,k)
      \end{pmatrix},\quad &k \in \overline{V}_3\cap\overline{V}_4.\\
    \end{cases}
  \end{align}
  By (\ref{deltasymm}) and (\ref{rjarjrsymm}), $v(\zeta,t,k)$ satisfies the symmetry
  \begin{equation}\label{vsymm}
    v(\zeta,t,k)=\overline{v(\zeta,t,-\overline{k})},\qquad \zeta\in\mathcal{I},\quad
    t>0,\quad k\in\Sigma.
  \end{equation}

  {\bf Step 3}
  Let $m^X(\zeta,z)$ be the solution of the RH problem in the
  complex $z$-plane:
  \begin{equation}\label{mXrhp}
    \begin{cases}
      m^X(\zeta,z)_+=m^X(\zeta,z)_-v^X(\zeta,z), \qquad z \in X,\\
      m^X(\zeta,z) \to I, \qquad z \to \infty,
    \end{cases}
  \end{equation}
  where contour $X=X_1\cup\cdots\cup X_4$ is shown in Figure \ref{X.pdf}.
  \begin{figure}
    \begin{center}
      \begin{overpic}[width=.7\textwidth]{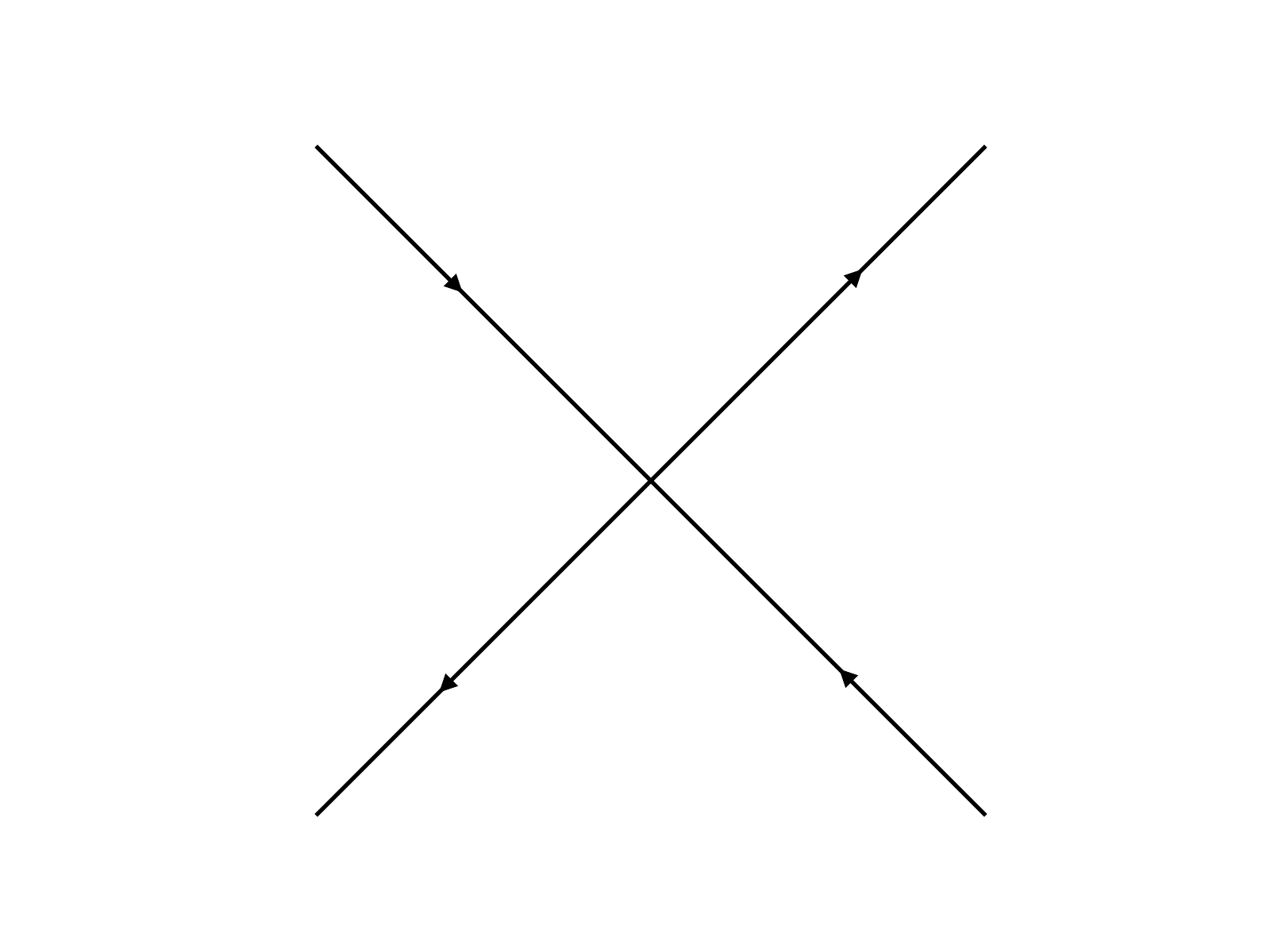}
        \put(62,54){$X_1$}
        \put(62,18){$X_2$}
        \put(36,18){$X_3$}
        \put(36,54){$X_4$}
        \put(50.5,32){0}
        \end{overpic}
      \begin{figuretext}\label{X.pdf}
        Contour $X$ in the complex $z$-plane
      \end{figuretext}
    \end{center}
  \end{figure}

  The jump matrix is
  \begin{align}\label{vX}
    v^X(\zeta,z)=
    \begin{cases}
      \begin{pmatrix}
        1 & 0\\
        q_1(\zeta)z^{-2i\nu(\zeta)}e^{\frac{iz^2}{2}} & 1
      \end{pmatrix},\qquad & z \in X_1,\\
      \begin{pmatrix}
        1 & q_2(\zeta)z^{2i\nu(\zeta)}e^{-\frac{iz^2}{2}}\\
        0 & 1
      \end{pmatrix},\qquad & z \in X_2,\\
      \begin{pmatrix}
        1 & 0\\
        -\frac{q_1(\zeta)}{q_1(\zeta)q_2(\zeta)}z^{-2i\nu(\zeta)}e^{\frac{iz^2}{2}} & 1
      \end{pmatrix},\qquad & z \in X_3,\\
      \begin{pmatrix}
        1 & -\frac{q_2(\zeta)}{q_1(\zeta)q_2(\zeta)}z^{2i\nu(\zeta)}e^{-\frac{iz^2}{2}}\\
        0 & 1
      \end{pmatrix},\qquad & z \in X_4.
    \end{cases}
  \end{align}
  We point out that in the local case (see \cite{L}), $q_1$ and $q_2$
  are defined as
  \begin{subequations}\label{q}
    \begin{align}\nonumber
      q_1(\zeta)&=\delta(\zeta,k)^{-2}r(k)\biggl(\frac{k-k_0}{\sqrt{48k_0}}\biggr)^{2i\nu(\zeta)}\Bigg|_{k=k_0}\\
      &=e^{-\chi(\zeta,k_0)}r(k_0)e^{2i\nu(\zeta)\ln(2\sqrt{48}k_0^{3/2})},\\
      q_2(\zeta)&=\overline{q_1(\zeta)}.
    \end{align}
  \end{subequations}
  By contrast, to keep the blondness of the function
  $q_j(\zeta),\zeta\in\mathcal{I},j=1,2$ for the nonlocal case,
  we let
  \begin{subequations}\label{q1q2}
    \begin{align}\nonumber
      q_1(\zeta)&=\delta(\zeta,k)^{-2}r_1(k)\biggl(\frac{2(k-k_0)}{k_0}\biggr)^{2i\nu(\zeta)}\Bigg|_{k=k_0}\\
      &=e^{-\chi(\zeta,k_0)}r_1(k_0)e^{2i\nu(\zeta)\ln4},\\\nonumber
      q_2(\zeta)&=\delta(\zeta,k)^2r_2(k)\biggl(\frac{2(k-k_0)}{k_0}\biggr)^{-2i\nu(\zeta)}\Bigg|_{k=k_0}\\
      &=e^{\chi(\zeta,k_0)}r_2(k_0)e^{-2i\nu(\zeta)\ln4}.
    \end{align}
  \end{subequations}
  From Lemma \ref{a1}, the unique
  solution $m^X(\zeta,z)$ of the RH problem (\ref{mXrhp}) can be
  explicitly expressed in terms of parabolic-cylinder function.
  Together with $D(\zeta,t)$ defined by
  \begin{equation}\label{D}
    D(\zeta,t)=e^{-\frac{t\phi(\zeta,0)}{2}\sigma_3}\tau^{-\frac{i\nu(\zeta)}{2}\sigma_3},
  \end{equation}
  where
  \begin{align}\nonumber
    \epsilon=\frac{k_0}{2}, \qquad \rho=\epsilon\sqrt{48k_0}, \qquad
    \tau=t\rho^2=12k_0^3t,\\\label{eps,rho,tau,phi}
    \phi(\zeta,z)=\Phi\biggl(\zeta,k_0+\frac{\epsilon}{\rho}z\biggr)
    =-16ik_0^3+\frac{iz^2}{2}+\frac{iz^3}{12\rho},
  \end{align}
  we use $m^X(\zeta,z)$ to introduce $m_0(\zeta,t,k)$ for $k$ near
  $k_0$:
  \begin{equation}\label{m0}
    m_0(\zeta,t,k)=D(\zeta,t)m^X\bigl(\zeta,\frac{\sqrt{\tau}}{\epsilon}(k-k_0)\bigr)D(\zeta,t)^{-1}, \qquad |k-k_0|\leq\epsilon,
  \end{equation}
  and extend it to a neighborhood of $-k_0$ by symmetry:
  \begin{equation}\label{m0-k0}
    m_0(\zeta,t,k)=\overline{m_0(\zeta,t,-\overline{k})}, \qquad |k+k_0|\leq\epsilon.
  \end{equation}

  \begin{remark}\label{Ddiff}
    The method introduced in \cite{L} cannot be imitated
    indiscriminately to deal with the situation where $\im \nu(\zeta) \neq 0$.
    In \cite{L} $D(\zeta,t)$ is defined by
    \begin{equation}\label{localD}
      D(\zeta,t)=e^{-\frac{t\phi(\zeta,0)}{2}\sigma_3}t^{-\frac{i\nu(\zeta)}{2}\sigma_3}.
    \end{equation}
    Notice (\ref{D}), we replace $t$ by $\tau$ in (\ref{localD})
    to define $D(\zeta,t)$. Actually, our adjustment including that to
    the function $q_j(\zeta),j=1,2$ (see (\ref{q})
    and (\ref{q1q2})) is also valid for the study of the local case.
  \end{remark}

  Then we use $m_0(\zeta,t,k)$ to introduce function $\hat{m}(\zeta,t,k)$:
  \begin{align}\label{hatm}
    \hat{m}(\zeta,t,k)=
    \begin{cases}
      m(\zeta,t,k)m_0(\zeta,t,k)^{-1}, \qquad & |k\pm k_0|<\epsilon,\\
      m(\zeta,t,k), \qquad & elsewhere.
    \end{cases}
  \end{align}
  By the RH problems (\ref{m}) and the definition of $m_0$,
  $\hat{m}(\zeta,t,k)$ satisfies the following RH problem
  \begin{equation}\label{hatmrhp}
    \begin{cases}
      \hat{m}(\zeta,t,k)_+=\hat{m}(\zeta,t,k)_-\hat{v}(\zeta,t,k), \qquad k \in \hat{\Gamma},\\
      \hat{m}(\zeta,t,k) \to I, \qquad k \to \infty,
    \end{cases}
  \end{equation}
  where $\hat{\Gamma}=\Gamma \cup \{k \, | \, |k \pm k_0| = \epsilon \}$
  is oriented as in Figure \ref{hatGamma.pdf},
  \begin{figure}
    \begin{center}
      \begin{overpic}[width=.7\textwidth]{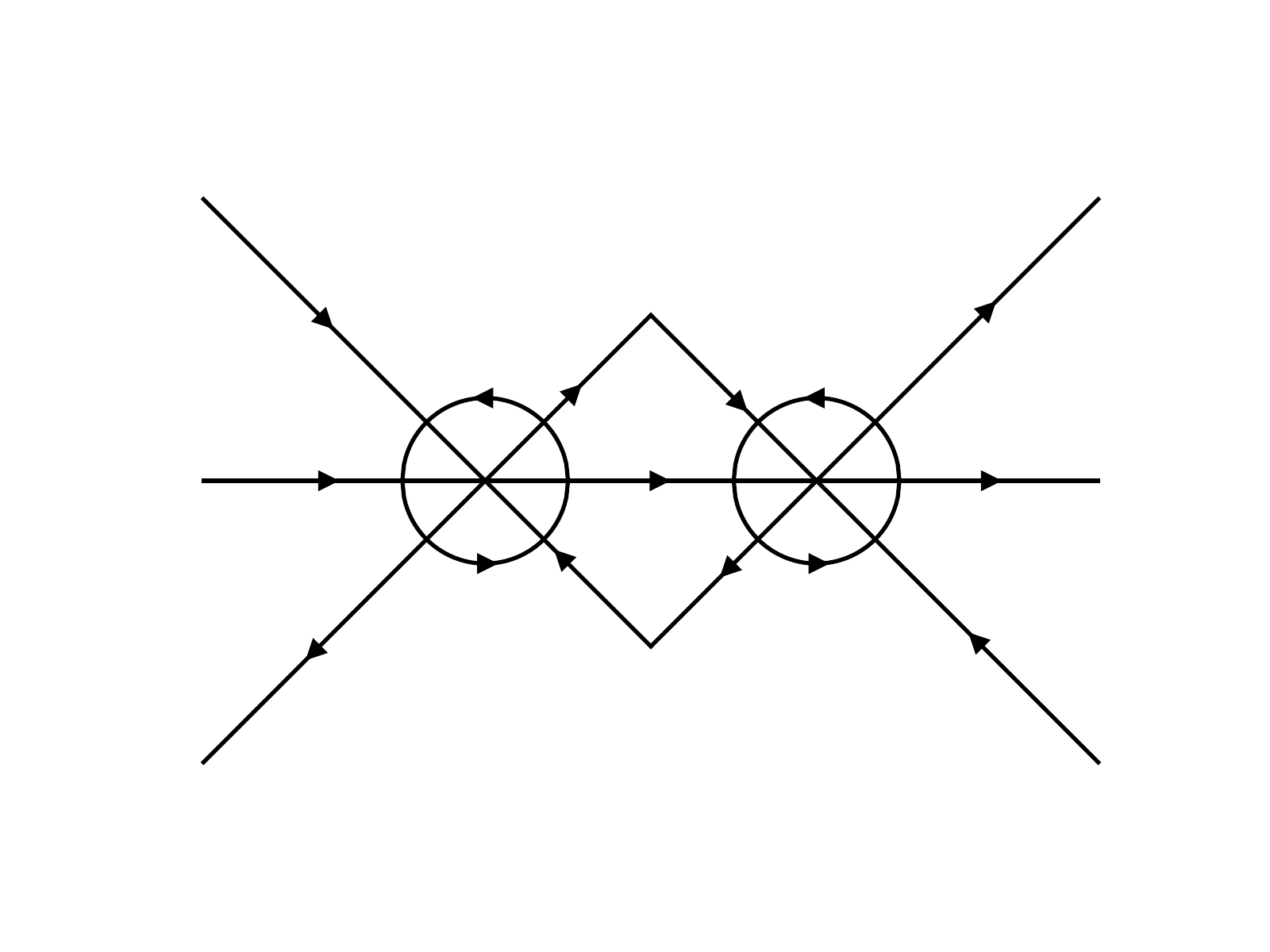}
        \put(35,32.5){$-k_0$}
        \put(63,32.5){$k_0$}
        \put(90,36){$\hat{\Gamma}$}
      \end{overpic}
      \begin{figuretext}\label{hatGamma.pdf}
        The jump contour $\hat{\Gamma}$
      \end{figuretext}
    \end{center}
  \end{figure}
  and the jump matrix is
  \begin{align}\label{hatGamma}
    \hat{v}(\zeta,t,k)=
    \begin{cases}
      m_{0-}(\zeta,t,k)v(\zeta,t,k)m_{0+}(\zeta,t,k)^{-1}, & \qquad |k_0\pm k_0|<\epsilon,\\
      m_0(\zeta,t,k), & \qquad |k_0\pm k_0|=\epsilon,\\
      v(\zeta,t,k), & \qquad elsewhere.
    \end{cases}
  \end{align}
  The model RH problem (\ref{hatmrhp}) is finally obtained.

  \section{Long time asymptotics}\nequation\label{lta}
  We use the model RH problem (\ref{hatmrhp}) to derive the asymptotics
  of the nonlocal mKdV equation in the similarity sector.

  Let $\Sigma$ denote the counter $\Sigma=\Sigma_1 \cup \cdots \cup \Sigma_6 \subset \C$,
  where
  \begin{align} \nonumber
    &\Sigma_1 = \bigl\{se^{\frac{i\pi}{4}}\, \big| \, 0 \leq s < \infty\bigr\}, &&
    \Sigma_2 = \bigl\{se^{-\frac{i\pi}{4}}\, \big| \, 0 \leq s < \infty\bigr\},
      \\ \nonumber
    &\Sigma_3 = \bigl\{se^{-\frac{3i\pi}{4}}\, \big| \, 0 \leq s < \infty\bigr\}, &&
    \Sigma_4 = \bigl\{se^{\frac{3i\pi}{4}}\, \big| \, 0 \leq s < \infty\bigr\},
      \\ \label{Sigmadef}
    &\Sigma_5 = \bigl\{s\, \big| \, 0 \leq s < \infty\bigr\}, &&
    \Sigma_6 = \bigl\{-s\, \big| \, 0 \leq s < \infty\bigr\},
  \end{align}
  are oriented as in Figure \ref{Sigma.pdf}. For $r>0$, we denote
  $\Sigma^r=\Sigma_1^r\cup\cdots\cup\Sigma_6^r$, where
  $\Sigma_j^r=\Sigma_j \cap D(0,r),\,j=1,\cdots,6$ and $D(k,r)$ is the disk of
  radius $r$ centered at $k$.

  \begin{figure}
    \begin{center}
      \begin{overpic}[width=.75\textwidth]{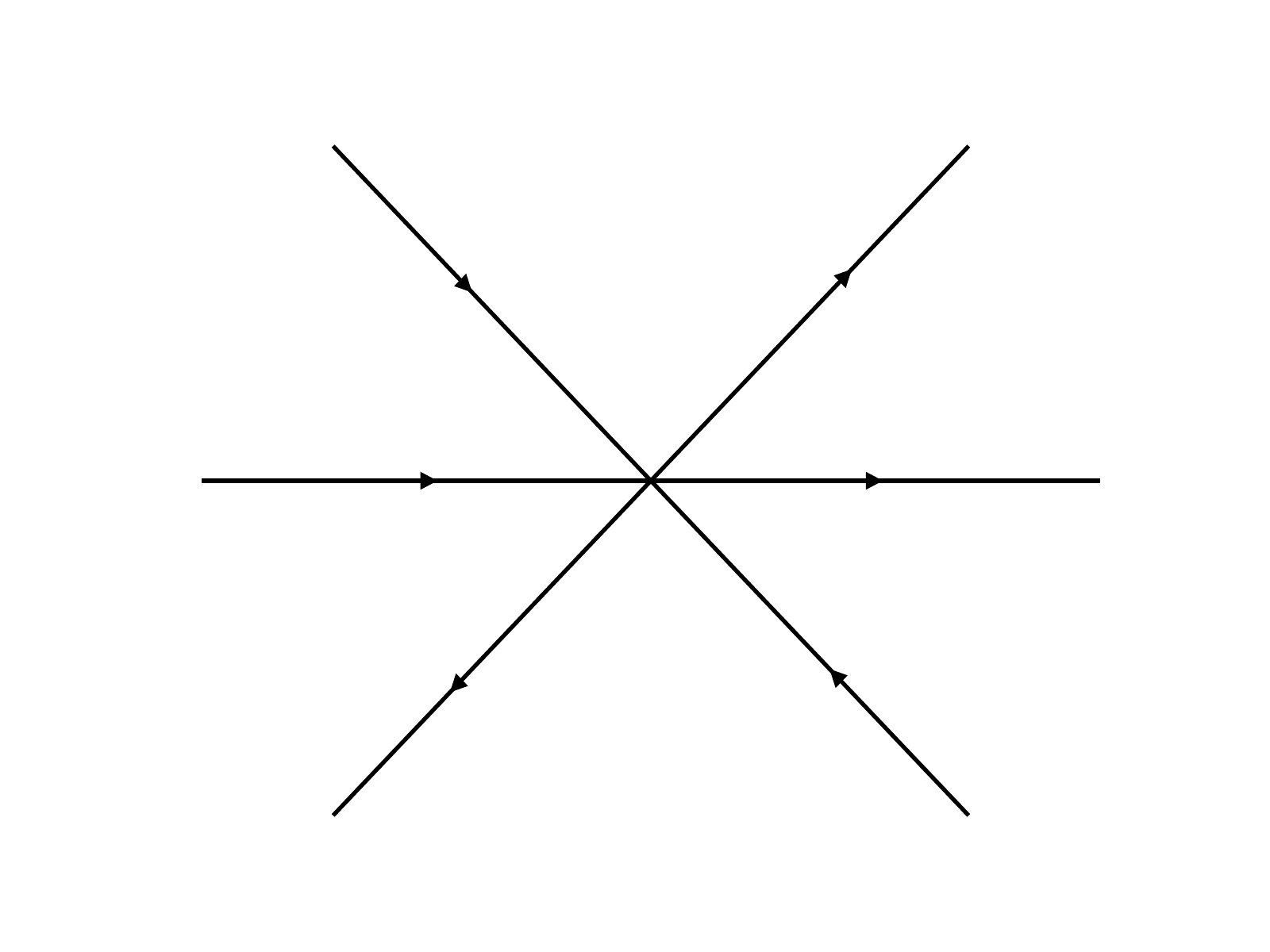}
        \put(64,58){$\Sigma_1$}
        \put(64,15){$\Sigma_2$}
        \put(34,15){$\Sigma_3$}
        \put(34,58){$\Sigma_4$}
        \put(74,39){$\Sigma_5$}
        \put(25,39){$\Sigma_6$}
      \end{overpic}
      \bigskip
      \begin{figuretext}\label{Sigma.pdf}
        The contour $\Sigma = \Sigma_1 \cup \cdots \cup \Sigma_6$.
      \end{figuretext}
    \end{center}
  \end{figure}

  \begin{lemma}\label{hatwlponGammaminusGammaSigma}
    Let $\Gamma_\Sigma=\pm k_0+\Sigma^\epsilon$
    and let $\Gamma'=\Gamma\setminus\Gamma_\Sigma$ as shown in Figure
    \ref{Gamma_minus_Sigma.pdf}.
    Let $\hat{w}(\zeta,t,k)=\hat{v}(\zeta,t,k)-I$, then
    $\hat{w}(\zeta,t,k)$ satisfies
    \begin{align}\label{wlpongammaminusgammasigma}
      \begin{cases}
        \Vert \hat{w}(\zeta,t,\cdot) \Vert_{L^p(\Gamma')}=O(\epsilon^{\frac{1}{p}}
        \tau^{-1}),\qquad p=1,2,\\
        \Vert \hat{w}(\zeta,t,\cdot) \Vert_{L^{\infty}(\Gamma')}=O(\tau^{-1}),
      \end{cases}
    \end{align}
    uniformly with respect to $\zeta \in \mathcal{I}$, as $\tau \to \infty$.
  \end{lemma}

  \begin{figure}
    \begin{center}
      \begin{overpic}[width=.75\textwidth]{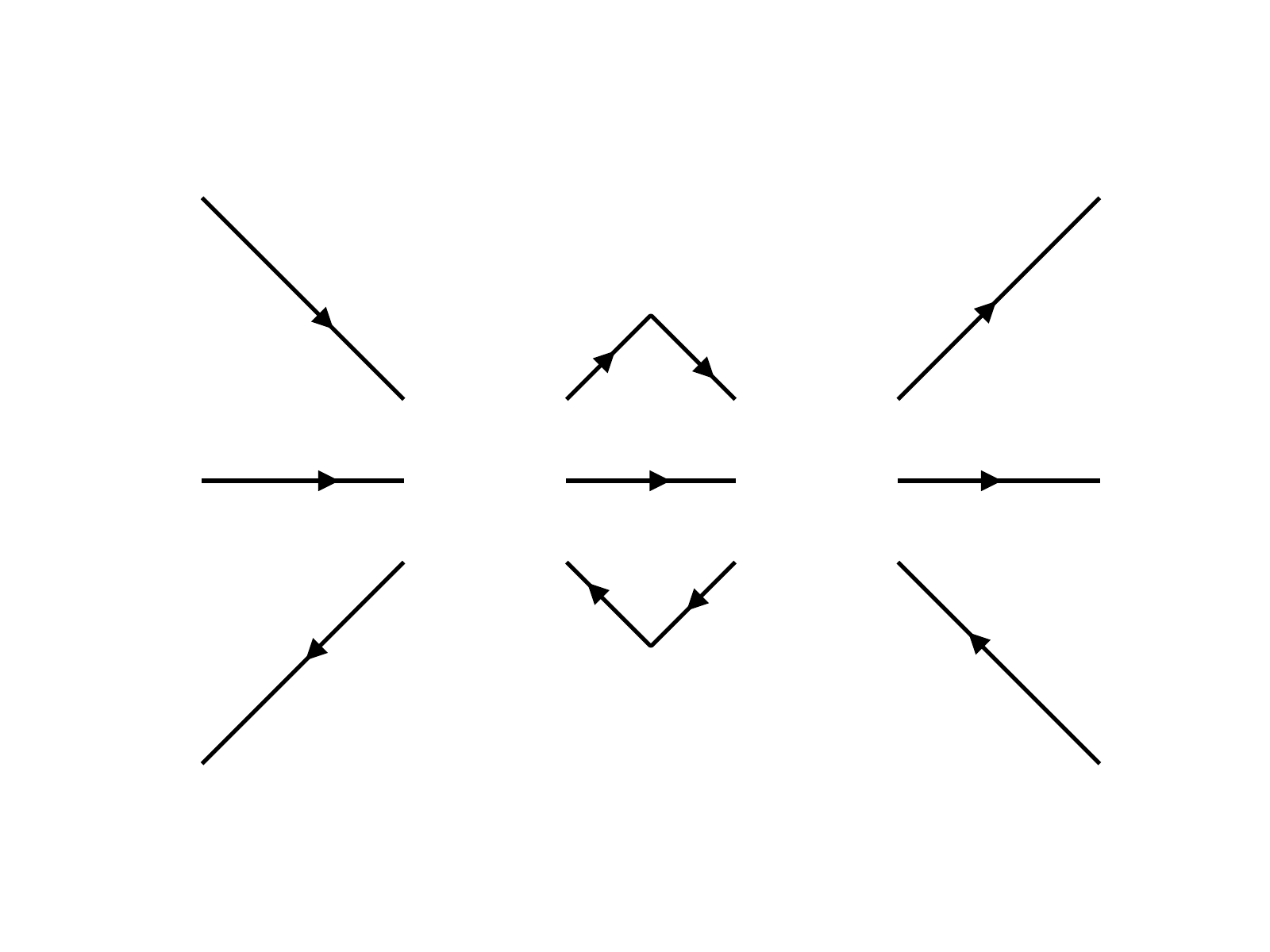}
        \put(90,36){$\Gamma'$}
      \end{overpic}
      \bigskip
      \begin{figuretext}\label{Gamma_minus_Sigma.pdf}
        The contour $\Gamma' = \Gamma \setminus \Gamma_{\Sigma^\epsilon}$.
      \end{figuretext}
    \end{center}
  \end{figure}

  \begin{proof}
    Let $\gamma$ denote the intersection of $\Gamma'$ and the line
    $k_0+\R e^{\frac{i\pi}{4}}$, i.e.
    \begin{equation}\label{gamma}
      \gamma = \Bigl\{k_0+ue^{\frac{i\pi}{4}}\Big|u\in
      \Bigl(-\sqrt{2}k_0, -\frac{k_0}{2}\Bigr]\cup
      \Big[\frac{k_0}{2}, \infty\Big)\Bigr\}.
    \end{equation}
    Let $k=k_0+ue^{\frac{i\pi}{4}}$.
    By (\ref{hatGamma}) and (\ref{v}), $\hat{w}$ has the following form on $\gamma$:
    \begin{align}\label{hatwongamma}
      \hat{w}(\zeta,t,k)=\begin{cases}
        \begin{pmatrix}
          0 & 0\\
          \delta(\zeta,k)^{-2}r_{1,a}(x,t,k)e^{t\Phi(\zeta,k)} & 0
        \end{pmatrix}, \qquad &\frac{k_0}{2}\leq u<\infty,\\
        \begin{pmatrix}
          0 & 0\\
          -\delta(\zeta,k)^{-2}r_{3,a}(x,t,k)e^{t\Phi(\zeta,k)} & 0
        \end{pmatrix}, \qquad &-\sqrt{2}k_0<u\leq -\frac{k_0}{2}.
      \end{cases}
    \end{align}
    It's enough to prove that $\delta(\zeta,k)^{\pm1}$ is uniformly bounded
    on $\Gamma'$ with respect to $\zeta \in \mathcal{I}$, i.e.
    \begin{equation}\label{supdeltaleq}
      \sup_{\zeta \in \mathcal{I}}\sup_{k \in \Gamma'} \bigl|\delta(\zeta,k)^{\pm1}\bigr| \leq C
    \end{equation}
    From Lemma \ref{chiuniformlybounded} $\chi(\zeta, k)$ is
    uniformly bounded on $\gamma$ with respect to
    $\zeta \in \mathcal{I}$. Thus
    \begin{align}\label{deltaleq}\nonumber
      \bigl|\delta(\zeta,k)^{\pm1}\bigr| & =\Biggl|\biggl(\frac{k-k_0}{k+k_0}\biggr)^{\pm i\nu}e^{\pm\chi(\zeta,k)}\Biggr|
      \leq C \Biggl|\biggl(\frac{ue^{i\pi/4}}{ue^{i\pi/4}+2k_0}\biggr)^{\pm i\nu}\Biggr|\\
      & \leq C \biggl|\Bigl(1+\frac{2k_0}{u}e^{-i\pi/4}\Bigr)^{\mp i\nu}\biggr|,
    \end{align}
    where $1+\frac{2k_0}{u}e^{-i\pi/4}$ satisfies the following inequalities
    \begin{equation}\label{1+2k0/ue-ipi/4}
      \begin{cases}
        1<\bigl|1+\frac{2k_0}{u}e^{-i\pi/4}\bigr|\leq(17+4\sqrt{2})^{\frac{1}{2}}, \qquad u \in \bigl[\frac{k_0}{2},\infty\bigr),\\
        1<\bigl|1+\frac{2k_0}{u}e^{-i\pi/4}\bigr|\leq(17-4\sqrt{2})^{\frac{1}{2}}, \qquad u \in \bigl(-\sqrt{2}k_0,-\frac{k_0}{2}\bigr].
      \end{cases}
    \end{equation}
    By (\ref{deltaleq}) and (\ref{1+2k0/ue-ipi/4}), $\delta(\zeta,k)^{\pm1}$
    is uniformly bounded on $\gamma$ with respect to
    $\zeta \in \mathcal{I}$. Since similar arguments apply to the
    remaining parts of $\Gamma'$, this prove (\ref{supdeltaleq}).

    Since the decompositions of $r_j$
    ($r_j=r_{j,a}+r_{j,r},j=1,\cdots,4$) is similar to the local case,
    we can follow \cite{L} to accomplish the rest of the proof by
    (\ref{rja-rjk0}), (\ref{r1ar2a}), (\ref{r1rr2r}) and
    (\ref{r3rr4r}).
  \end{proof}

  We normalize the jump matrix $v(\zeta,t,k)$ on $\Gamma_{\Sigma}$
  \begin{equation}\label{v0}
    v_0(\zeta,t,z)= v\biggl(\zeta,t,k_0+\frac{\epsilon z}{\rho}\biggr), \qquad z \in \Sigma^{\rho},
  \end{equation}
  which has the form of
  \begin{equation}\label{normalizedv0}
    v_0(\zeta,t,z)=
    \begin{cases}
      \begin{pmatrix}
        1 & 0\\
        R_1(\zeta,t,z)\bigl(\frac{z}{\rho}\bigr)^{-2i\nu(\zeta)}e^{t\phi(\zeta,z)} & 1
      \end{pmatrix}, &\qquad z \in \Sigma_1^{\rho},\\
      \begin{pmatrix}
        1 & R_2(\zeta,t,z)\bigl(\frac{z}{\rho}\bigr)^{2i\nu(\zeta)}e^{-t\phi(\zeta,z)}\\
        0 & 1
      \end{pmatrix}, &\qquad z \in \Sigma_2^{\rho},\\
      \begin{pmatrix}
        1 & 0\\
        -R_3(\zeta,t,z)\bigl(\frac{z}{\rho}\bigr)^{-2i\nu(\zeta)}e^{t\phi(\zeta,z)} & 1
      \end{pmatrix}, &\qquad z \in \Sigma_3^{\rho},\\
      \begin{pmatrix}
        1 & -R_4(\zeta,t,z)\bigl(\frac{z}{\rho}\bigr)^{2i\nu(\zeta)}e^{-t\phi(\zeta,z)}\\
        0 & 1
      \end{pmatrix}, &\qquad z \in \Sigma_4^{\rho},\\
      \begin{pmatrix}
        1-S_1(\zeta,t,z)S_2(\zeta,t,z) & -S_2(\zeta,t,z)\bigl(\frac{z}{\rho}\bigr)^{2i\nu(\zeta)}e^{-t\phi(\zeta,z)}\\
        S_1(\zeta,t,z)\bigl(\frac{z}{\rho}\bigr)^{-2i\nu(\zeta)}e^{t\phi(\zeta,z)} & 1
      \end{pmatrix}, &\qquad z \in \Sigma_5^{\rho},\\
      \begin{pmatrix}
        1 & -S_4(\zeta,t,z)\bigl(\frac{z}{\rho}\bigr)^{2i\nu(\zeta)}e^{-t\phi(\zeta,z)}\\
        S_3(\zeta,t,z)\bigl(\frac{z}{\rho}\bigr)^{-2i\nu(\zeta)}e^{t\phi(\zeta,z)} & 1-S_3(\zeta,t,z)S_4(\zeta,t,z)
      \end{pmatrix}, &\qquad z \in \Sigma_6^{\rho}.\\
    \end{cases}
  \end{equation}
  where $\epsilon$, $\rho$ and $\phi$ are defined by (\ref{eps,rho,tau,phi}).
  The phase $\phi(\zeta,z)$ is identical to the local case,
  which is a smooth function of $(\zeta,z) \in \mathcal{I} \times \C$
  satisfying condition (2.10) and (2.11) of Theorem 2.1 in \cite{L}.
  Moreover, $\{R_j(\zeta,t,k)\}_1^4$ and $\{S_j(\zeta,t,k)\}_1^4$
  satisfy the following Lemma.
  \begin{lemma}\label{RjSj}
    There exist constants $(\alpha,L) \in [1/2,1) \times (0,\infty)$
    such that the functions $\{R_j(\zeta,t,z)\}_1^4$ and
    $\{S_j(\zeta,t,k)\}_1^4$ satisfy the inequalities:
    \begin{align}\label{lipschitz}
      \begin{cases}
        |R_1(\zeta,t,z)-q_1(\zeta)| \leq L\bigl| \frac{z}{\rho} \bigr|^{\alpha}e^{\frac{t|z|^2}{6}}, \qquad &z \in \Sigma_1^\rho,\\
        |R_2(\zeta,t,z)-q_2(\zeta)| \leq L\bigl| \frac{z}{\rho} \bigr|^{\alpha}e^{\frac{t|z|^2}{6}}, \qquad &z \in \Sigma_2^\rho,\\
        \bigl|R_3(\zeta,t,z)-\frac{q_1(\zeta)}{1-q_1(\zeta)q_2(\zeta)}\bigr| \leq L\bigl| \frac{z}{\rho} \bigr|^{\alpha}e^{\frac{t|z|^2}{6}}, \qquad &z \in \Sigma_3^\rho,\\
        \bigl|R_4(\zeta,t,z)-\frac{q_2(\zeta)}{1-q_1(\zeta)q_2(\zeta)}\bigr| \leq L\bigl| \frac{z}{\rho} \bigr|^{\alpha}e^{\frac{t|z|^2}{6}}, \qquad &z \in \Sigma_4^\rho,\\
        |S_j(\zeta,t,z)| \leq L\big|\frac{z}{\rho}\big|t^{-\frac{3}{2}}, \qquad \qquad \qquad j=1,2, &z \in \Sigma_5^\rho,\quad \\
        |S_j(\zeta,t,z)| \leq L\big|\frac{z}{\rho}\big|t^{-\frac{3}{2}}, \qquad \qquad \qquad j=3,4, &z \in \Sigma_6^\rho,\quad
      \end{cases}\, \zeta \in \mathcal{I}, \quad t>0.
    \end{align}
    where $q_1(\zeta)$ and $q_2(\zeta)$ are defined by (\ref{q1q2}).
  \end{lemma}

  \begin{proof}
    Equations (\ref{v}) and (\ref{normalizedv0}) imply that
    \begin{equation*}
      \begin{cases}
        R_1(\zeta,t,z)=\delta(\zeta,k)^{-2}r_{1,a}(x,t,k)\bigl(\frac{z}{\rho}\bigr)^{2i\nu(\zeta)},\quad
        R_2(\zeta,t,z)=\delta(\zeta,k)^{2}r_{2,a}(x,t,k)\bigl(\frac{z}{\rho}\bigr)^{-2i\nu(\zeta)},\\
        R_3(\zeta,t,z)=\delta(\zeta,k)^{-2}r_{3,a}(x,t,k)\bigl(\frac{z}{\rho}\bigr)^{2i\nu(\zeta)},\quad
        R_4(\zeta,t,z)=\delta(\zeta,k)^{2}r_{4,a}(x,t,k)\bigl(\frac{z}{\rho}\bigr)^{-2i\nu(\zeta)},\\
        S_1(\zeta,t,z)=\delta(\zeta,k)^{-2}r_{1,r}(x,t,k)\bigl(\frac{z}{\rho}\bigr)^{2i\nu(\zeta)},\quad
        S_2(\zeta,t,z)=\delta(\zeta,k)^{2}r_{2,r}(x,t,k)\bigl(\frac{z}{\rho}\bigr)^{-2i\nu(\zeta)},\\
        S_3(\zeta,t,z)=\delta(\zeta,k)^{-2}r_{3,r}(x,t,k)\bigl(\frac{z}{\rho}\bigr)^{2i\nu(\zeta)},\quad
        S_4(\zeta,t,z)=\delta(\zeta,k)^{2}r_{4,r}(x,t,k)\bigl(\frac{z}{\rho}\bigr)^{-2i\nu(\zeta)}.
      \end{cases}
    \end{equation*}
    Let $k=k_0+\frac{\epsilon}{\rho}z$. Using the expression
    (\ref{deltachi}), we can write
    \begin{equation}\label{R1}
      R_1(\zeta,t,z)=e^{-2\chi(\zeta,k)}r_{1,a}(x,t,k)e^{2i\nu(\zeta)\ln\bigl(\frac{2(k+k_0)}{k_0}\bigr)},\qquad z\in \Sigma_1^{\rho}.
    \end{equation}
    By (\ref{rja-rjk0}), we find
    \begin{align}\label{R10}\nonumber
      R_1(\zeta,t,0)&=e^{-2\chi(\zeta,k_0)}r_1(k_0)e^{2i\nu(\zeta)\ln4}\\
      &=q_1(\zeta).
    \end{align}
    Similar arguments imply that
    \begin{equation*}
      R_2(\zeta,t,0)=q_2(\zeta), \qquad R_3(\zeta,t,0)=\frac{q_1(\zeta)}{1-q_1(\zeta)q_2(\zeta)},
      \qquad R_4(\zeta,t,0)=\frac{q_2(\zeta)}{1-q_1(\zeta)q_2(\zeta)}.
    \end{equation*}
    Note that $z\in\Sigma_1^\rho$ is equivalent to $k\in k_0+\Sigma_1^\epsilon$.
    For $z\in\Sigma_1^\rho$, we have
    \begin{align}\label{R1-q1}\nonumber
      |R_1(\zeta,t,z)-q_1(\zeta)| & \leq \Bigl|e^{-2\chi(\zeta,k)}
      -e^{-2\chi(\zeta,k_0)}\Bigr| \biggl|r_{1,a}(x,t,k)
      e^{2i\nu(\zeta)\ln\bigl(\frac{2(k+k_0)}{k_0}\bigr)}\biggr|\\\nonumber
      & + \Bigl|e^{-2\chi(\zeta,k_0)}\Bigr| |r_{1,a}(x,t,k)-r_1(k_0)|
      \biggl|e^{2i\nu(\zeta)\ln\bigl(\frac{2(k+k_0)}{k_0}\bigr)}\biggr|\\
      & + \Bigl|e^{-2\chi(\zeta,k_0)}r_1(k_0)\Bigr|
      \biggl|1-e^{-2i\nu(\zeta)\ln\bigl(\frac{k+k_0}{2k_0}\bigr)}\biggr|
      \biggl|e^{2i\nu(\zeta)\ln\bigl(\frac{2(k+k_0)}{k_0}\bigr)}\biggr|.
    \end{align}
    From Lemma \ref{chiuniformlybounded}, $e^{2\chi(\zeta,k_0)}$ is
    bounded for $\zeta\in\mathcal{I}$. Moreover, let
    $k=k_0+ue^{i\pi/4}$, we have
    \begin{align*}
      \biggl|e^{2i\nu(\zeta)\ln\bigl(\frac{2(k+k_0)}{k_0}\bigr)}\biggr|=
      \Bigl|\Bigl(4+\frac{2u}{k_0}e^{i\pi/4}\Bigr)^{2i\nu(\zeta)}\Bigr|,
    \end{align*}
    where $4+\frac{2u}{k_0}e^{i\pi/4}$ satisfies
    \begin{align*}
      (17-4\sqrt{2})^{\frac{1}{2}}<\Bigl|4+\frac{2u}{k_0}e^{i\pi/4}\Bigr|\leq4,
      \qquad 0\leq u<\epsilon.
    \end{align*}
    Thus $e^{2i\ln\bigl(\frac{2(k+k_0)}{k_0}\bigr)}$ is uniformly bounded
    with respect to $\zeta\in\mathcal{I}$ and $k\in k_0+\Sigma_1^\epsilon$.

    As in the local case, using (\ref{rja-rjk0}), (\ref{r1ar2a})
    and the estimate
    \begin{equation*}
      |\re\phi(\zeta,ve^{i\pi/4})|\leq\frac{2v^2}{3}, \qquad 0\leq v<\rho,
    \end{equation*}
    we can prove that for $z\in\Sigma_1^\rho$, the following
    inequalities holds:
    \begin{align}\label{r1a-r10}
      &|r_{1,a}(x,t,k)-r_1(k_0)| \leq C\frac{\epsilon|z|}{\rho}e^{\frac{t}{6}|z|^2},\\
      \label{r1}
      &|r_{1,a}(x,t,k)|\leq Ce^{\frac{t}{6}|z|^2}.
    \end{align}
    Thus,
    \begin{align}\nonumber
      |R_1(\zeta,t,z)-q_1(\zeta)|&\leq Ce^{\frac{t|z|^2}{6}}
      \Bigl|e^{-2\chi(\zeta,k)}-e^{-2\chi(\zeta,k_0)}\Bigr|+
      C\frac{\epsilon|z|}{\rho}e^{\frac{t|z|^2}{6}}\\
      \label{R_1-q_1}
      &+C\biggl|1-e^{-2i\nu(\zeta)\ln\bigl(\frac{k+k_0}{2k_0}\bigr)}\biggr|.
    \end{align}
    Employing (\ref{chi}), we have
    \begin{equation}\label{chik-chik0}
      |\chi(\zeta,k)-\chi(\zeta,k_0)|\leq
      C\Big|\ln\big(\frac{k_0+k}{2k_0}\big)\Big|+
      C\Big|\int_{-k_0}^{k_0}\ln\big(\frac{s-k}{s-k_0}\big)d
      \ln (1-r_1(s)r_2(s))\Big|.
    \end{equation}
    Then we can follow the proof of Lemma 4.9 in \cite{L}
    to derive the following inequalities
    \begin{align}\label{lemma4.92}
      &\biggl|1-e^{-2i\nu(\zeta)\ln\bigl(\frac{k+k_0}{2k_0}\bigr)}\biggr|
      \leq Ck_0^{-1}|k-k_0|,\\\label{lemma4.91}
      &\Bigl|e^{-2\chi(\zeta,k)}-e^{-2\chi(\zeta,k_0)}\Bigr|
      \leq C|k-k_0|(k_0^{-1}+|\ln|k-k_0||).
    \end{align}
    Notice that (\ref{lemma4.91}) is slightly different from
    Lemma 4.9 in \cite{L}, which is caused by the first term of
    right-hand side of (\ref{chik-chik0}).

    Using (\ref{R_1-q_1}), (\ref{lemma4.92}) and (\ref{lemma4.91}),
    We can verify (\ref{lipschitz}) in the case of $z\in\Sigma_1^\rho$;
    the case of $z\in\Sigma_j^\rho,\,j=2,3,4$ are similar.
    Using (\ref{r1rr2r}), (\ref{r3rr4r}) and the uniform boundness of
    $e^{-2\chi(\zeta,k)}$ and $e^{2i\ln\bigl(\frac{2(k+k_0)}{k_0}\bigr)}$,
    it's clear that $S_j(\zeta,t,z),\,j=1,\cdots,4$ satisfy the inequalities in
    (\ref{lipschitz}).
  \end{proof}

  \begin{lemma}\label{hatwestimateonGammaSigma}
    On $\Gamma_{\Sigma}$, function $\hat{w}$ satisfies
    \begin{align}\nonumber
      \hat{w}(\zeta,t,k)=O\bigg(\tau^{-\frac{\alpha}{2}+
      |\im \nu(\zeta)|}&e^{-\frac{\tau}{24\epsilon^2}|k\mp k_0|^2}
      \bigg),\\\label{hatwleq1}
      &\tau \to \infty, \quad \zeta \in \mathcal{I}, \quad k \in
      \pm k_0+\displaystyle\bigcup_{j=1}^{4}\Sigma_j^\epsilon,
    \end{align}
    and
    \begin{align}\label{hatwleq2}
      \hat{w}(\zeta,t,k)=O\bigg(\tau^{-\frac{3}{2}+2|\im \nu(\zeta)|}\bigg),
      \qquad
      \tau \to \infty, \quad \zeta \in \mathcal{I}, \quad k \in
      \pm k_0+\displaystyle\bigcup_{j=5}^{6}\Sigma_j^\epsilon,
    \end{align}
    where the error term is uniform with respect to $\zeta\in\mathcal
    {I}$ and $k\in \Gamma_{\Sigma}$.
  \end{lemma}

  \begin{proof}
    We prove the case of $k\in k_0+\Sigma^\epsilon$. Symmetries (\ref{vsymm})
    and (\ref{m0-k0}) implies that
    \begin{equation}\label{hatwsymm}
      \hat{w}(\zeta,t,k)=\overline{\hat{w}(\zeta,t,-\overline{k})},\qquad
      \zeta\in\mathcal{I},\quad t>0, \quad k\in\hat{\Gamma}.
    \end{equation}
    Thus the case of
    $k\in -k_0+\Sigma^\epsilon$ follows by the above symmetry.
    We let $k=k_0+\frac{\epsilon}{\rho}z,\, z\in\Sigma^\rho$. Then
    \begin{align*}
      \hat{w}(\zeta,t,k)&=m_{0-}(\zeta,t,k)v(\zeta,t,k)m_{0+}(\zeta,
      t,k)^{-1}-I\\
      &=D(\zeta,t)m_-^X(\zeta,\sqrt{t}z)D(\zeta,t)^{-1}v_0(\zeta,t,z)
      D(\zeta,t)m_+^X(\zeta,\sqrt{t}z)^{-1}D(\zeta,t)^{-1}-I
    \end{align*}
    By Lemma \ref{a1}, there exists a constant $G$ such that
    $m_\pm^X(\zeta,\sqrt{t}z)$ is uniformly
    bounded with respect to $\zeta\in\mathcal{I}$ and
    $|\sqrt{t}z|\geq G$;
    $m_\pm^X (\zeta,\sqrt{t}z)(\sqrt{t}z)^{i\nu(\zeta)\sigma_3}$
    is uniformly bounded with respect to $\zeta\in\mathcal{I}$ and
    $|\sqrt{t}z| < G$. Thus we write $\hat{w}$ as
    \begin{align}\label{hatw=<G>=G}
      \hat{w}(\zeta,t,k)=
      \begin{cases}
        D(\zeta,t)m_-^X(\zeta,\sqrt{t}z)u_1(\zeta,t,z)
        m_+^X(\zeta,\sqrt{t}z)^{-1}D(\zeta,t)^{-1},
        \qquad &|\sqrt{t}z|\geq G,\\
        D(\zeta,t)m_-^X(\zeta,\sqrt{t}z)\bigg((\sqrt{t}z)
        ^{i\nu(\zeta)\hat{\sigma_3}}u_2(\zeta,t,z)\bigg)
        m_+^X(\zeta,\sqrt{t}z)^{-1}
        D(\zeta,t)^{-1}, &|\sqrt{t}z| < G.
      \end{cases}
    \end{align}
    where
    \begin{align}\label{u1}
      &u_1(\zeta,t,z)=D(\zeta,t)^{-1}v_0(\zeta,t,z)D(\zeta,t)
      -v^X(\zeta,\sqrt{t}z)\\\label{u2}
      &u_2(\zeta,t,z)=(\sqrt{t}z)^{-i\nu(\zeta)\hat{\sigma_3}}
      u_1(\zeta,t,z)
    \end{align}
    By (\ref{D}), $|D(\zeta,t)|=O(\tau^{\frac{|\im\nu(\zeta)|}{2}})$.
    Consequently, it's enough to prove that
    \begin{align}\label{u1estimate1}
      u_1(\zeta,t,z)=O\Big(\tau^{-\frac{\alpha}{2}}e^
      {-\frac{t|z|^2}{24}}\Big),\qquad \tau\to\infty,\quad \zeta\in
      \mathcal{I},\quad z\in\displaystyle\bigcup_{j=1}^4\Sigma_j^\rho,
      \quad |\sqrt{t}z|\geq G,\\\label{u1estimate2}
      u_1(\zeta,t,z)=O\Big(\tau^{-\frac{3}{2}+|\im\nu(\zeta)|}\Big),
      \qquad \tau\to\infty,\quad \zeta\in\mathcal{I},
      \quad z\in\displaystyle\bigcup_{j=5}^6\Sigma_j^\rho,
      \quad |\sqrt{t}z|\geq G,\\
      \label{u2estimate1}
      u_2(\zeta,t,z)=O\Big(\tau^{-\frac{\alpha}{2}}e^
      {-\frac{t|z|^2}{24}}\Big),\qquad \tau\to\infty,\quad \zeta\in
      \mathcal{I},\quad z\in\displaystyle\bigcup_{j=1}^4\Sigma_j^\rho,
      \quad |\sqrt{t}z|<G,\\\label{u2estimate2}
      u_2(\zeta,t,z)=O\Big(\tau^{-\frac{3}{2}}
      \Big),\qquad \tau\to\infty,\quad \zeta\in
      \mathcal{I},\quad z\in\displaystyle\bigcup_{j=5}^6\Sigma_j^\rho,
      \quad |\sqrt{t}z|<G
    \end{align}
    uniformly with respect to $(\zeta,z)$ in the given ranges.

    For the case of $z\in\Sigma_1^\rho$, we have
    \begin{equation}\label{u1onSigma1}
      u_1(\zeta,t,z)=
      \begin{pmatrix}
        0 & 0\\
        \Big(R_1(\zeta,t,z)e^{t(\phi(\zeta,z)-\phi(\zeta,0))}
        -q_1(\zeta)e^{\frac{itz^2}{2}}\Big)(\sqrt{t}z)^{-2i\nu(\zeta)}
        & 0
      \end{pmatrix}.
    \end{equation}
    So only the (21) entry of $u_1(\zeta,t,z)$ is nonzero, and for
    $|\sqrt{t}z|\geq G$ we find that
    \begin{align}\label{(u1)21eq}\nonumber
      \big|\big(u_1(\zeta,t,z)\big)_{21}\big|&=\Big|R_1(\zeta,t,z)
      e^{t(\phi(\zeta,z)-\phi(\zeta,0))}-q_1(\zeta)
      e^{\frac{itz^2}{2}}\Big|\Big|(\sqrt{t}z)^{-2i\nu(\zeta)}\Big|\\
      \nonumber
      &=\Big|R_1(\zeta,t,z)e^{t\hat{\phi}(\zeta,z)}-q_1(\zeta)\Big|
      e^{-\frac{t|z|^2}{2}}\big|\sqrt{t}z\big|^{2\im\nu(\zeta)}
      e^{\frac{\pi}{2}\re\nu(\zeta)}\\
      &\leq C\bigg(|R_1(\zeta,t,z)-q_1(\zeta)|e^{t\re\hat{\phi}
      (\zeta,z)}+|q(\zeta)|\Big|e^{t\hat{\phi}(\zeta,z)}-1\Big|
      \bigg)e^{-\frac{t|z|^2}{2}}\big|\sqrt{t}z\big|^{2\im\nu(\zeta)},
    \end{align}
    where $\hat{\phi}(\zeta,z)=\phi(\zeta,z)-\phi(\zeta,0)
    -\frac{iz^2}{2}$. As in the local case\cite{L}, we can use the
    inequalities
    \begin{align*}
      &|e^w-1|\leq|w|\max(1,e^{\re w}), \qquad w\in\C,\\
      &\re\hat{\phi}(\zeta,z)\leq \frac{|z|^2}{4},\qquad
      \zeta\in\mathcal{I},\quad z\in\Sigma_1^\rho,
    \end{align*}
    and the boundness of $q_j(\zeta),\,j=1,2$ to find that
    \begin{align}\nonumber
      \big|\big(u_1(\zeta,t,z)\big)_{21}\big|\leq
      C\big(|R_1(\zeta,t,z)-q_1(\zeta)|+Ct&|\hat{\phi}(\zeta,z)|
      \big)e^{-\frac{t|z|^2}{4}}\big|\sqrt{t}z\big|^{2\im\nu(\zeta)},\\
      \label{u1leq}
      &\zeta\in\mathcal{I},\quad t>0,\quad z\in\Sigma_1^\rho,
      \quad |\sqrt{t}z|\geq G.
    \end{align}
    It's easy to verify that
    \begin{align*}
      |\hat{\phi}(\zeta,z)|\leq C\frac{|z|^3}{\rho},\qquad
      \zeta\in\mathcal{I},\quad z\in\Sigma^\rho,
    \end{align*}
    so together with Lemma \ref{RjSj}, the right-hand of
    (\ref{u1leq}) is of order
    \begin{align}\nonumber
      &O\bigg(\bigg(\frac{L|z|^\alpha e^{\frac{t|z|^2}{6}}}
      {\rho^\alpha}+\frac{Ct|z|^3}{\rho}\bigg)e^{-\frac{t|z|^2}{4}}
      (t|z|^2)^{\im\nu(\zeta)}\bigg)\\\nonumber
      &=O\bigg(\bigg(\frac{(t|z|^2)
      ^{\alpha/2+\im \nu(\zeta)}}{\tau^{\alpha/2}}+
      \frac{(t|z|^2)^{3/2+\im \nu(\zeta)}}{\tau^{1/2}}
      \bigg)e^{-\frac{t|z|^2}{12}}\bigg)\\
      &=O\bigg(\bigg(\frac{1}{\tau^{\alpha/2}}+\frac{1}{\tau^{1/2}}
      \bigg)e^{-\frac{t|z|^2}{24}}\bigg),\qquad\tau\to\infty,
      \quad \zeta\in\mathcal{I}, \quad z\in\Sigma_1^\rho, \quad |\sqrt{t}z|\geq G.
    \end{align}
    uniformly with respect to $(\zeta,z)$ in the given ranges.
    This proves (\ref{u1estimate1}) for $z\in\Sigma_1^\rho$; the
    cases of $z\in\Sigma_j^\rho,\,j=2,3,4,$ are similar.
    For the case of $z\in\Sigma_5^\rho$, we have
    \begin{align}\label{u1onSigma5}
      u_1(\zeta,t,z)=
      \begin{pmatrix}
        -S_1(\zeta,t,z)S_2(\zeta,t,z) & -S_2(\zeta,t,z)e^{-t(\phi(\zeta,z)-\phi(\zeta,0))}
        (\sqrt{t}z)^{2i\nu(\zeta)}\\
        S_1(\zeta,t,z)e^{t(\phi(\zeta,z)-\phi(\zeta,0))}(\sqrt{t}z)^{-2i\nu(\zeta)}
        & 0
      \end{pmatrix}.
    \end{align}
    By Lemma \ref{RjSj}, the (11) entry satisfies
    \begin{align}\nonumber
      \big|(u_1(\zeta,t,z))_{11}\big|=&|S_1(\zeta,t,z)S_2(\zeta,t,z)|
      \leq C\Big|\frac{z}{\rho}\Big|^2t^{-3}\\\label{u111onSigma5}
      &\leq C\frac{|z|^2}{\tau}t^{-2} \leq C\tau^{-3},\qquad
      \zeta\in\mathcal{I},\quad t>0, \quad z\in\Sigma_5^\rho,
    \end{align}
    and the (12) entry satisfies
    \begin{align}\nonumber
      \big|(u_1(\zeta,t,z))_{12}\big|&=\bigg|S_2(\zeta,t,z)e^{-t(\phi(\zeta,z)-
      \phi(\zeta,0))}(\sqrt{t}z)^{2i\nu(\zeta)}\bigg|\\\nonumber
      &\leq C\frac{|z|^{1-2\im\nu(\zeta)}}{\rho}t^{-\frac{3}{2}-\im\nu(\zeta)}
      \leq C\frac{|z|^{1-2\im\nu(\zeta)}}{\tau^{\frac{1}{2}}}t^{-1-\im\nu(\zeta)}.
    \end{align}
    Because $-\frac{1}{2}<\im\nu(\zeta)<\frac{1}{2}$, $(u_1(\zeta,t,z))_{12}$
    is of order
    \begin{equation}\label{u112onSigma5}
      (u_1(\zeta,t,z))_{12}=O\Big(\tau^{-\frac{3}{2}-\im\nu(\zeta)}\Big),
      \qquad \tau\to\infty, \quad  \zeta\in\mathcal{I}, \quad z\in\Sigma_5^\rho.
    \end{equation}
    Similarly, $(u_1(\zeta,t,z))_{21}$ is of order
    \begin{equation}\label{u121onSigma5}
      (u_1(\zeta,t,z))_{21}=O\Big(\tau^{-\frac{3}{2}+\im\nu(\zeta)}\Big),
      \qquad \tau\to\infty, \quad  \zeta\in\mathcal{I}, \quad z\in\Sigma_5^\rho.
    \end{equation}
    Using (\ref{u111onSigma5}), (\ref{u112onSigma5}) and (\ref{u121onSigma5}),
    we prove (\ref{u1estimate2}) for $z\in\Sigma_5^\rho$; the case of $z\in\Sigma_6^\rho$
    is similar.

    On the other hand, (\ref{u2estimate1}) and (\ref{u2estimate2}),
    the estimates of $u_2(\zeta,t,z)$, can be proved in the same way.
  \end{proof}

  Following Lemma 2.6 in \cite{L}, we use Lemma \ref{hatwlponGammaminusGammaSigma}
  and Lemma \ref{hatwestimateonGammaSigma} to obtain the estimates:
  \begin{align}\label{hatwl2}
    \Vert\hat{w}(\zeta,t,\cdot)\Vert_{L^2(\hat{\Gamma})}=O(\epsilon^{\frac{1}{2}}\tau
    ^{-\frac{\alpha}{2}+|\im\nu(\zeta)|}),\qquad \tau\to\infty,\quad
    \zeta\in\mathcal{I},\\\label{hatwlinfty}
    \Vert\hat{w}(\zeta,t,\cdot)\Vert_{L^\infty(\hat{\Gamma})}=O(\tau
    ^{-\frac{\alpha}{2}+|\im\nu(\zeta)|}),\qquad \tau\to\infty,\quad
    \zeta\in\mathcal{I},\\\label{hatwlp}
    \Vert\hat{w}(\zeta,t,\cdot)\Vert_{L^p(\pm k_0+\Sigma^\epsilon)}=
    O(\epsilon^{\frac{1}{p}}\tau
    ^{-\frac{1}{2p}-\frac{\alpha}{2}+|\im\nu(\zeta)|}),\qquad \tau\to\infty,\quad
    \zeta\in\mathcal{I},
  \end{align}
  where $p\in[1,\infty)$ and the error terms are uniform with respect to
  $\zeta\in\mathcal{I}$. Moreover, if taking acocunt of the first and second columns
  of $\hat{w}(\zeta,t,k)$ respectively in Lemma \ref{hatwestimateonGammaSigma}, we have
  \begin{align}\label{hatwjl2}
    \Vert\hat{w}^{(j)}(\zeta,t,\cdot)\Vert_{L^2(\hat{\Gamma})}=
    O(\epsilon^{\frac{1}{2}}\tau
    ^{-\frac{\alpha}{2}+(-1)^j\im\nu(\zeta)}),\qquad \tau\to\infty,\quad
    \zeta\in\mathcal{I},\quad j=1,2\\\label{hatwjlinfty}
    \Vert\hat{w}^{(j)}(\zeta,t,\cdot)\Vert_{L^\infty(\hat{\Gamma})}=O(\tau
    ^{-\frac{\alpha}{2}+(-1)^j\im\nu(\zeta)}),\qquad \tau\to\infty,\quad
    \zeta\in\mathcal{I},\quad j=1,2\\\label{hatwjlp}
    \Vert\hat{w}^{(j)}(\zeta,t,\cdot)\Vert_{L^p(\pm k_0+\Sigma^\epsilon)}=
    O(\epsilon^{\frac{1}{p}}\tau
    ^{-\frac{1}{2p}-\frac{\alpha}{2}+(-1)^j\im\nu(\zeta)}),\qquad \tau\to\infty,\quad
    \zeta\in\mathcal{I},\quad j=1,2
  \end{align}

  \begin{lemma}\label{rhpasymptotics}
    % Assume that
    % \begin{equation}\label{imnurestriction}
    %   -\frac{1}{4}<\im\nu(\zeta)<\frac{1}{4},\qquad \zeta\in\mathcal{I}.
    % \end{equation}
    The RH problem (\ref{hatmrhp}) has a unique solution for all
    sufficiently large $\tau$. And for any %$\sup_{\zeta\in\mathcal{I}}|\im\nu(\zeta)|
    $\alpha\in(\lambda,1)$ this solution satisfies
    \begin{align}\nonumber
      \lim_{k\to\infty}(k\hat{m}(\zeta,t,k))_{12}=-\frac{2i\epsilon\re\beta(\zeta,t)}
      {\tau^{\frac{1}{2}-\im\nu(\zeta)}}+O\Big(\epsilon\tau^{-\frac{1+\alpha}{2}
      +|\im\nu(\zeta)|+\im \nu(\zeta)}\Big)\\\label{rhpasymp}
      \tau\to\infty, \quad \zeta\in\mathcal{I},
    \end{align}
    where $\lambda=\max\Big(\frac{1}{2},\sup\limits_{\zeta\in\mathcal{I}}2|\im\nu(\zeta)|\Big)$,
    and the error term is uniform with respect to $\zeta\in\mathcal{I}$ and
    $\beta(\zeta,t)$ is defined by
    \begin{equation}\label{beta}
      \beta(\zeta,t)=\frac{\sqrt{2\pi}e^{i\pi/4}e^{-\pi\nu(\zeta)/2}}{q_1(\zeta)
      \Gamma(-i\nu(\zeta))}e^{-t\phi(\zeta,0)}\tau^{-i\re\nu(\zeta)}.
    \end{equation}
  \end{lemma}

  \begin{proof}
    We define the integral operator $\hat{\mathcal{C}}_{\hat{w}}:L^2(\hat{\Gamma})+
    L^{\infty}(\hat{\Gamma})\to L^2(\hat{\Gamma})$ by $\hat{\mathcal{C}}_{\hat{w}}f=
    \hat{\mathcal{C}}_-(f\hat{w})$, where $\hat{\mathcal{C}}_-(f\hat{w})$ is the
    boundary value of $\hat{\mathcal{C}}(f\hat{w})$ from the right side of
    $\hat{\Gamma}$, and $\hat{\mathcal{C}}$ is the Cauchy operator associated with
    $\hat{\Gamma}$:
    \begin{equation*}
      (\hat{\mathcal{C}}f)(z)=\frac{1}{2\pi i}\int_{\hat{\Gamma}}\frac{f(s)}{s-z}ds,
      \qquad z\in\C\setminus\hat{\Gamma}.
    \end{equation*}
    By (\ref{hatwlinfty}),
    \begin{equation}\label{cwbl2}
      \Vert \hat{\mathcal{C}}_{\hat{w}} \Vert_{\mathcal{B}(L^2(\hat{\Gamma}))}\leq
      C\Vert\hat{w}\Vert_{L^\infty(\hat{\Gamma})}=O(\tau^{-\frac{\alpha}{2}+
      |\im\nu(\zeta)|}), \qquad \tau\to\infty,\quad \zeta\in\mathcal{I}.
    \end{equation}
    Since $\alpha\in(\lambda,1)$,
    $\Vert \hat{\mathcal{C}}_{\hat{w}} \Vert_{\mathcal{B}(L^2(\hat{\Gamma}))}$
    decays to 0 as $\tau\to\infty$. Thus, there exists a $T>0$ such that
    $I-\hat{\mathcal{C}}_{\hat{w}(\zeta,t,\cdot)}\in\mathcal{B}(L^2(\hat{\Gamma}))$
    is invertible for all $(\zeta,t)\in(0,\infty)$ with $\tau>T$.

    Moreover, by (\ref{hatwl2}) we have
    \begin{equation}\label{mu-i}
      \Vert \hat{\mu}-I \Vert_{L^2(\hat{\Gamma})}=O(\epsilon^{\frac{1}{2}}
      \tau^{-\frac{\alpha}{2}+|\im\nu(\zeta)|}), \qquad \tau\to\infty,
      \quad \zeta\in\mathcal{I}.
    \end{equation}
    where $\hat{\mu}-I=(I-\hat{\mathcal{C}}_{\hat{w}})^{-1}\hat{\mathcal{C}}_{\hat{w}}I
    \in L^2(\hat{\Gamma})$ is the solution of the integral equation
    \begin{equation*}
      (I-\hat{\mathcal{C}}_{\hat{w}})(\mu-I)=\hat{\mathcal{C}}_{\hat{w}}I.
    \end{equation*}

    Consequently, by Lemma 2.9 in \cite{L},
    there exists a unique solution $\hat{m}$ of the RH problem (\ref{hatmrhp})
    whenever $\tau>T$, and in accordance with the local case, we can represent $\hat{m}$
    as
    \begin{equation}\label{hatmsolution}
      \hat{m}(\zeta,t,k)=I+\hat{\mathcal{C}}(\hat{\mu}\hat{w})=I+\frac{1}{2\pi i}
      \int_{\hat{\Gamma}}\hat{\mu}(\zeta,t,s)\hat{w}(\zeta,t,s)\frac{ds}{s-k}.
    \end{equation}
    Finally, Lemma 2.10 in \cite{L} and symmetry (\ref{hatwsymm}) imply that
    \begin{align}\nonumber
      &\lim_{k\to\infty}k(\hat{m}(\zeta,t,k)-I)=-\frac{1}{2\pi i}\int_{\hat{\Gamma}}
      \hat{\mu}(\zeta,t,k)\hat{w}(\zeta,t,k)dk\\\nonumber
      &=-\frac{1}{2\pi i}\bigg(\int_{|k-k_0|=\epsilon}+\int_{|k+k_0|=\epsilon}\bigg)
      \hat{\mu}(\zeta,t,k)\hat{w}(\zeta,t,k)dk-\frac{1}{2\pi i}\int_\Gamma
      \hat{\mu}(\zeta,t,k)\hat{w}(\zeta,t,k)dk\\\label{limktoinftyk(m-i)}
      &=-\frac{1}{\pi i}\re \bigg(\int_{|k-k_0|=\epsilon}\hat{\mu}(\zeta,t,k)
      (m_0(\zeta,t,k)^{-1}-I)dk\bigg)-\frac{1}{2\pi i}\int_\Gamma\hat{\mu}(\zeta,t,k)
      \hat{w}(\zeta,t,k)dk.
    \end{align}
    By Lemma \ref{a1}, we have
    \begin{align}\nonumber
      \bigg(m_0(\zeta,t,k)^{-1}\bigg)^{(2)}&=\bigg(D(\zeta,t)m^X\Big(\zeta,
      \frac{\sqrt{\tau}}{\epsilon}(k-k_0)\Big)^{-1}D(\zeta,t)^{-1}\bigg)^{(2)}\\
      \label{m0-1}
      &=
      \begin{pmatrix}
        0\\1
      \end{pmatrix}
      +\frac{B^{(2)}(\zeta,t)}{\sqrt\tau(k-k_0)}+O(\tau^{-1+\im\nu(\zeta)}),\qquad
      \tau\to\infty,\quad \zeta\in\mathcal{I},\quad |k-k_0|=\epsilon,
    \end{align}
    where $B(\zeta,t)$ is defined by
    \begin{equation*}
      B(\zeta,t)=-i\epsilon
      \begin{pmatrix}
        0 & -\beta^X(\zeta)e^{-t\phi(\zeta,0)}\tau^{-i\nu(\zeta)}\\
        \gamma^X(\zeta)e^{t\phi(\zeta,0)}\tau^{i\nu(\zeta)} & 0
      \end{pmatrix}.
    \end{equation*}
    Using (\ref{mu-i}) and (\ref{m0-1}) we find
    \begin{align}\nonumber
      &\int_{|k-k_0|=\epsilon}\bigg(\hat{\mu}(\zeta,t,k)(m_0(\zeta,t,k)^{-1}-I)\bigg)
      ^{(2)}dk=
      \int_{|k-k_0|=\epsilon}(m_0(\zeta,t,k)^{-1}-I)^{(2)}dk\\\nonumber
      &+\int_{|k-k_0|=\epsilon}(\hat{\mu}(\zeta,t,k)-I)(m_0(\zeta,t,k)^{-1}-I)^{(2)}dk
      \\\nonumber
      &=\frac{B^{(2)}(\zeta,t)}{\sqrt\tau}\int_{|k-k_0|=\epsilon}\frac{dk}{k-k_0}+
      O(\epsilon\tau^{-1+\im\nu(\zeta)})+O(\Vert\hat{\mu}(\zeta,t,\cdot)-I\Vert
      _{L^2(\hat{\Gamma})}\epsilon^{\frac{1}{2}}\tau^{-\frac{1}{2}+\im\nu(\zeta)})\\
      \nonumber
      &=\frac{2\pi iB^{(2)}(\zeta,t)}{\sqrt\tau}+O(\epsilon\tau^{-1+\im\nu(\zeta)})
      +O(\epsilon\tau^{-\frac{1+\alpha}{2}+|\im\nu(\zeta)|+\im\nu(\zeta)})\\
      \label{int_|k-k0|=epsilon}
      &=\frac{2\pi iB^{(2)}(\zeta,t)}{\sqrt\tau}
      +O(\epsilon\tau^{-\frac{1+\alpha}{2}+|\im\nu(\zeta)|+\im\nu(\zeta)}),
      \qquad \tau\to\infty,\quad \zeta\in\mathcal{I},
    \end{align}
    uniformly with respect to $\zeta\in\mathcal{I}$.
    Notice that $B^{(2)}(\zeta,t)$ contains $\tau^{-i\nu(\zeta)}$, the order of the
    leading term of (\ref{int_|k-k0|=epsilon}) is $\tau^{-\frac{1}{2}+\im\nu(\zeta)}$.
    Since $\alpha\in(\lambda,1)$, the error term of (\ref{int_|k-k0|=epsilon})
    does make sense compared to the leading term.

    On the other hand,
    \begin{align*}
      \bigg|\int_{\Gamma}\bigg(\hat{\mu}(\zeta,t,k)\hat{w}(\zeta,t,k)\bigg)^{(2)}dk\bigg|&=
      \bigg|\int_{\Gamma}(\hat{\mu}(\zeta,t,k)-I)\hat{w}^{(2)}(\zeta,t,k)dk+\int_{\Gamma}
      \hat{w}^{(2)}(\zeta,t,k)dk\bigg|\\
      &\leq \Vert\hat{\mu}-I\Vert_{L^2(\Gamma)}\Vert\hat{w}^{(2)}\Vert_{L^2(\Gamma)}
      +\Vert\hat{w}^{(2)}\Vert_{L^1(\Gamma)}.
    \end{align*}
    % The $L^1$-norm of $\hat{w}^{(2)}$ is $O(\epsilon\tau^{-1})$ on $\Gamma'$ by
    % (\ref{wlpongammaminusgammasigma}) and is
    % $O(\epsilon\tau^{-\frac{1+\alpha}{2}+\im\nu(\zeta)})$ on
    % $\pm k_0+\Sigma^{\epsilon}$ by (\ref{hatwjlp}). Thus
    (\ref{wlpongammaminusgammasigma}) and (\ref{hatwjlp}) implies that
    $\Vert\hat{w}^{(2)}\Vert_{L^1(\Gamma)}=O(\epsilon\tau^{-1}+
    \epsilon\tau^{-\frac{1+\alpha}{2}+\im\nu(\zeta)})$,
    and
    $\Vert\hat{w}^{(2)}\Vert_{L^2(\Gamma)}=O(\epsilon^{\frac{1}{2}}\tau^{-1}+
    \epsilon^{\frac{1}{2}}\tau^{-\frac{1}{4}-\frac{\alpha}{2}+\im\nu(\zeta)})$.
    Since $\Vert\hat{\mu}-I\Vert_{L^2(\Gamma)}=O(\epsilon^{\frac{1}{2}}
    \tau^{-\frac{\alpha}{2}+|\im\nu(\zeta)|})$ by (\ref{mu-i}) and $\alpha\in(\lambda,1)$,
    we find that
    \begin{align}\nonumber
      \bigg|\int_{\Gamma}\bigg(\hat{\mu}(\zeta,t,k)\hat{w}(\zeta,t,k)\bigg)^{(2)}dk\bigg|&=
      O(\epsilon\tau^{-1}+\epsilon\tau^{-\frac{1+\alpha}{2}+\im\nu(\zeta)}+\epsilon
      \tau^{-\frac{1}{4}-\alpha+|\im\nu(\zeta)|+\im\nu(\zeta)})\\
      \label{inthatmuhatwongamma}
      &=O(\epsilon\tau^{-\frac{1+\alpha}{2}+|\im\nu(\zeta)|+\im\nu(\zeta)}),
      \qquad \tau\to\infty, \quad \zeta\in\mathcal{I},
    \end{align}
    uniformly with respect to $\zeta\in\mathcal{I}$. Then (\ref{limktoinftyk(m-i)})
    (\ref{m0-1}), (\ref{int_|k-k0|=epsilon}) and (\ref{inthatmuhatwongamma}) imply
    (\ref{rhpasymp}).
  \end{proof}

  \begin{theorem}\label{nmkdvlongtime}
    Consider the Cauchy problem (\ref{nmkdv_ini}). We assume that the scattering data
    associated the initial data $q_0(x)$ are such that:
    \begin{enumerate}
      \item $a_1(k)$ and $a_2(k)$ have no zeros in $\overline{\C_+}$ and
      $\overline{\C_-}$ respectively;
      \item For $\zeta\in\mathcal{I}$, $\Delta(\zeta)=\int_{-\infty}^{k_0}
      d\arg(1-r_1(s)r_2(s))\in(-\pi,\pi)$, where $r_1(s)=\frac{b(s)}{a_1(s)}$ and
      $r_1(s)=\frac{b(s)}{a_2(s)}$.
    \end{enumerate}
    Then, for any $\alpha\in(\lambda,1)$ and $N>0$, the solution $q(x,t)$ defined by
    (\ref{ist4.35}) satisfies
    \begin{align}\nonumber
      q(x,t)=\frac{4\epsilon\re\beta(\zeta,t)}{\tau^{\frac{1}{2}-\im\nu(\zeta)}}+
      O(\epsilon\tau^{-\frac{1+\alpha}{2}+|\im\nu(\zeta)|+\im\nu(\zeta)}),\\
      \tau\to\infty,\,\, -Nt<x<0,
    \end{align}
    where the error term is uniform with respect to $x$ in the given range.
  \end{theorem}

  \begin{proof}
    Lemma \ref{rhpasymptotics} implies that (\ref{ist4.35}) exists for all sufficiently
    large $\tau$, and
    \begin{align}\nonumber
      q(x,t)&=2i\lim_{k\to\infty}(kM(x,t,k))_{12}=2i\lim_{k\to\infty}
      (k\hat{m}(x,t,k))_{12}\\
      &=\frac{4\epsilon\re\beta(\zeta,t)}{\tau^{\frac{1}{2}-\im\nu(\zeta)}}+
      O(\epsilon\tau^{-\frac{1+\alpha}{2}+|\im\nu(\zeta)|+\im\nu(\zeta)}).
    \end{align}
  \end{proof}

  \begin{remark}
    In contrast with the local mKdV equation, the decay rate of the leading term depends
    on $\zeta=\frac{x}{t}$ through $\im\nu(\zeta)$. Notice that in the local case,
    $\im\nu(\zeta)=0$ for all $\zeta\in\mathcal{I}$, and Theorem \ref{nmkdvlongtime}
    regresses to the main result of \cite{L}.
  \end{remark}

  \begin{remark}
    In section 4 of \cite{RS}, the conditions (1) and (2) in Theorem \ref{nmkdvlongtime} were
    verified in the case of single box initial data, for which the scattering data can
    be calculated explicitly.
  \end{remark}

  \appendix
  \section{}
  \renewcommand{\theequation}{A.\arabic{equation}}\nequation
  \begin{lemma}\label{a1}
    The RH problem (\ref{mXrhp}) has a unique solution $m^X(\zeta,z)$ for each $\zeta
    \in\mathcal{I}$. This solution satisfies
    \begin{equation}\label{mXasymp}
      m^X(\zeta,z)=I+\frac{i}{z}
      \begin{pmatrix}
        0 & -\beta^X(\zeta)\\
        \gamma^X(\zeta) & 0
      \end{pmatrix}
      +O(\frac{1}{z^2}), \qquad z\to\infty, \zeta\in\mathcal{I},
    \end{equation}
    where the error term is uniform with respect to $\arg z\in[0,2\pi]$ and $\zeta\in
    \mathcal{I}$. The functions $\beta^X(\zeta)$ and $\gamma^X(\zeta)$ are defined by
    \begin{align}\label{betaX}
      \beta^X(\zeta)=\frac{\sqrt{2\pi}e^{i\pi/4}e^{-\pi\nu(\zeta)/2}}{q_1(\zeta)
      \Gamma(-i\nu(\zeta))},\\\label{gammaX}
      \gamma^X(\zeta)=\frac{\sqrt{2\pi}e^{-i\pi/4}e^{-\pi\nu(\zeta)/2}}{q_2(\zeta)
      \Gamma(i\nu(\zeta))}.
    \end{align}
    Moreover, for each closed disk $K\in\C$ centered at the origin,
    \begin{equation}\label{supsupmXzinusigma3}
      \sup_{\zeta\in\mathcal{I}}\sup_{z\in K\setminus X}\Big|m^X(\zeta,z)
      z^{i\nu(\zeta)\sigma_3}\Big|<\infty.
    \end{equation}
  \end{lemma}

  \begin{proof}
    The detailed proofs can be found in Appendix A in \cite{RS} and Appendix B
    in \cite{L}. Notice that $m^X(\zeta,z)$ is singular at the origin, which is
    different from the local case. Multiplying $m^X(\zeta,z)$ by
    $z^{i\nu(\zeta)\sigma_3}$ can remove the singularity.
  \end{proof}

  \bibliographystyle{plain}

\begin{thebibliography}{99}
  \small
\baselineskip=18pt

  \bibitem{AC}
  M. J. Ablowitz and P. A. Clarkson,
  \textit{Solitons, Nonlinear Evolution Equations and Inverse Scattering},
  London Mathematical Society Lecture Note Series, vol. 149,
  Cambridge University Press, Cambridge, 1991.

  \bibitem{AKNS}
  M. J. Ablowitz, D. J. Kaup, A. C. Newell and H. Segur,
  The Inverse Scattering Transform‐Fourier Analysis for Nonlinear Problems,
  \textit{Studies in Applied Mathematics.}
  \textbf{53} (1974), 249-315.

  \bibitem{AM1}
  M.J. Ablowitz and Z.H. Musslimani,
  Integrable nonlocal nonlinear Schr\"odinger equation,
  \textit{Phys. Rev. Lett.} \textbf{110} 064105 (2013).

  \bibitem{AM2}
  M. J. Ablowitz and Z. H. Musslimani,
  Inverse Scattering transform for the integrable nonlocal nonlinear Schr\"odinger equation,
  \textit{Nonlinearity.} \textbf{29} (2016), 915--946.

  \bibitem{AM3}
  M. J. Ablowitz and Z. H. Musslimani,
  Integrable nonlocal nonlinear equations,
  \textit{Stud. Appl. Math.} \textbf{139} (2017), 7--59.

  \bibitem{AS}
  M. J. Ablowitz and H. Segur,
  \textit{Soliton and the Inverse Scattering Transform}, SIAM, Philadelphia, 1981.

  \bibitem{BC}
  R. Beals and R. R. Coifman,
  Scattering and inverse scattering for first order systems,
  \textit{Comm. Pure Appl. Math.} \textbf{37} (1987), 39--90.

  \bibitem{BDT}
  R. Beals, P. A. Deift, and C. Tomet,
  \textit{Direct and Inverse Scattering on the Line},
  Mathematical Surveys and Monographs, vol. 28, American Mathematical Society,
  Providence, RI, 1988.

\bibitem{bkst}
 A. Boutet de Monvel, A. Kostenko, D. Shepelsky, G. Teschl, {\em Long-time asymptotics for the Camassa¨CHolm equa-tion}, SIAM J. Math. Anal. {\bf 41} (2009) 1559-1588.



\bibitem{cvz}
P.J. Cheng, S. Venakides, X. Zhou,
{\em Long-time asymptotics for the pure radiation solution of the sine-Gordon equa-tion}, Comm. Partial Differential Equations {\bf 24} (1999) 1195-1262.

\bibitem{diz}
 P. Deift, A. R. Its, and X. Zhou, {\em Long-time asymptotics for integrable nonlinear wave equations, Important developments in soliton theory.} Springer Ser. Nonlinear Dynam., Springer, Berlin, (1993), 181-204.

  \bibitem{DZ}
  P. Deift and X. Zhou,
  A Steepest Descent Method for Oscillatory Riemann--Hilbert Problems. Asymptotics for the MKdV Equation,
  \textit{Annals of Mathematics.}
  \textbf{137} (1993), 295-368.

  \bibitem{FT}
  L. D. Faddeev and L. A. Takhtajan,
  \textit{Hamiltonian Methods in the Theory of Solitons},
  Reprint of the 1987 English edition, Classics in Mathematics, Springer, Berlin, 2007.
  Translated from the 1986 Russian original by Al.G. Reyman.

  \bibitem{GA}
  T. A. Gadzhimuradov and A. M. Agalarov,
  Towards a gauge-equivalent magnetic structure of the nonlocal nonlinear Schrodinger equation,
  \textit{Physical Review A.}
  \textbf{93} 062124 (2016).

  \bibitem{gt}
K. Grunert, G. Teschl, {\em Long-time asymptotics for the Korteweg¨Cde Vries equation via nonlinear steepest descent}, Math. Phys. Anal. Geom. {\bf 12} (2009) 287


  \bibitem{HL}   Huang Lin, Lenells Jonatan,
Nonlinear Fourier transforms for the sine-Gordon equation in the quarter plane, J.  differential equations, 264(2018), 3445-3499 


  \bibitem{IN}
  A. R. Its and V. Y. Novokshenov,
  \textit{The Isomonodromic Deformation Method in the Thoery of Painlev\'e Equation},
  Lecture Notes in Mathematics, Vol. 1191, Springer, 1986.

  \bibitem{JZ1}
  J. L. Ji and Z. N. Zhu,
  On a nonlocal modified Korteweg-de Vries equation: Integrability, Darboux transformation and soliton solutions,
  \textit{Studies in Applied Mathematics.}
  \textbf{42} (2017), 699-708.

  \bibitem{JZ2}
  J. L. Ji and Z. N. Zhu,
  Soliton solutions of an integrable nonlocal modified Korteweg-de Vries equation through inverse scattering transform,
  \textit{Journal of Mathematical Analysis and Applications.}
  \textbf{453} (2017), 973-984.

\bibitem{kv1}
 A.V. Kitaev, A.H. Vartanian, {\em Asymptotics of solutions to the modified nonlinear Schr$\ddot{o}$dinger equation: solution on a nonvanishing continuous background}, SIAM J. Math. Anal. {\bf 30} (1999) 787$-$832.

\bibitem{kv2}
A.V. Kitaev, A.H. Vartanian, {\em Higher order asymptotics of the modified non-linear Schr$\ddot{o}$dinger equation}, Comm. Par$-$324.tial Differential Equations {\bf 25} (2000) 1043-1098.


  \bibitem{L}
  J. Lenells,
  The Nonlinear Steepest Descent Method for Riemann-Hilbert Problems of Low Regularity,
  \textit{Indiana University Mathematics Journal.}
  \textbf{66} (2017), 1287-1332.

  \bibitem{LH}
  S. Y. Lou and F. Huang,
  Alice-Bob Physics: Coherent Solutions of Nonlocal KdV Systems,
  \textit{Scientific Reports.}
  \textbf{7} (2017), 869.

  \bibitem{MECM}
  K. G. Makris, R. El-Ganainy, D. N. Christodoulides and Z. H. Musslimani,
  Beam dynamics in PT symmetric optical lattices,
  \textit{Physical Review Letters.}
  \textbf{100}(2008), 103904.

  \bibitem{MMEC}
  Z. H. Musslimani, K. G. Makris, R. El-Ganainy and D. N. Christodoulides
  Optical solitons in PT periodic potentials,
  \textit{Physical Review Letters.}
  \textbf{100}(2008), 030402.

  \bibitem{RS}
  Y. Rybalko and D. Shepelsky,
  Long-time asymptotics for the integrable nonlocal nonlinear Schr\"odinger equation,
  preprint, available at https://arxiv.org/abs/1710.07961.

  \bibitem{TLH}
  X. Y. Tang, Z. F. Liang and X. Z. Hao,
  Nonlinear waves of a nonlocal modified KdV equation in the atmospheric and oceanic dynamical system,
  \textit{Communications in Nonlinear Science and Numerical Simulation.}
  \textbf{60} (2018), 62--71.

 \bibitem{X}
 J. Xu, Long-time asymptotics for the short pulse equation, J. Differential Equations, accepted


  \bibitem{XF}
  J. Xu and E. G. Fan, Long-time asymptotics for the Fokas-Lenells equation with decaying initial value problem: Without solitons, J. Differential Equations, 259(2015), 1098-1148.

 \bibitem{XF2}
  J. Xu and E. G. Fan, Long-time asymptotic behavior for the complex short pulse equation, arXiv:1712.07815.

    \bibitem{XFC}
  J. Xu, E. G. Fan and Y. Chen, Long-time Asymptotic for the Derivative Nonlinear
Schr¡§odinger Equation with Step-like Initial Value, Math Phys Anal Geom, 16(2013), 253-288.

  \bibitem{YY}
  B. Yang and J. Yang,
  Transformations between Nonlocal and Local Integrable Equations,
  \textit{Studies in Applied Mathematics.}
  \textbf{140} (2017), 178--201.

\bibitem{ZXF}
 Q. Z.  Zhu,  J. Xu and E. G. Fan, The Riemann¨CHilbert problem and long-time asymptotics for the
Kundu-Eckhaus equation with decaying initial value, Applied Mathematics Letters, 76 (2018) 81-89.






  \end{thebibliography}

  \end{document}